\documentclass[aps,pra,onecolumn,superscriptaddress,10pt]{revtex4}
\usepackage[letterpaper, margin=1in]{geometry}
\usepackage{tikzit}
\tikzstyle{dot}=[inner sep=0.3mm, minimum width=2mm, minimum height=2mm, draw, shape=circle, font={\footnotesize}, tikzit fill=magenta]
\tikzstyle{white dot}=[dot, fill=white, text depth=-0.2mm, tikzit category=ZH-pf, draw=black]
\tikzstyle{white phase dot}=[minimum size=5mm, font={\footnotesize\boldmath}, shape=rectangle, rounded corners=2mm, inner sep=0.2mm, outer sep=-2mm, scale=0.8, tikzit shape=circle, draw=black, fill=white, tikzit category=ZH-pf, tikzit fill=white, tikzit draw=blue]
\tikzstyle{gray dot}=[dot, fill={rgb,255: red,180; green,180; blue,180}, text depth=-0.2mm, tikzit category=ZH-pf]
\tikzstyle{gray phase dot}=[white phase dot, tikzit shape=circle, tikzit draw=blue, fill={rgb,255: red,180; green,180; blue,180}, font={\footnotesize\boldmath}]
\tikzstyle{hadamard}=[fill=white, draw, inner sep=0.6mm, minimum height=1.5mm, minimum width=1.5mm, shape=rectangle, tikzit shape=rectangle, tikzit category=ZH-pf]
\tikzstyle{small hadamard}=[hadamard]
\tikzstyle{lambda}=[hadamard, fill={rgb,255: red,180; green,180; blue,180}, tikzit shape=rectangle]
\tikzstyle{halfscalar}=[star, fill=black, draw=black, minimum size=8pt, inner sep=0pt]
\tikzstyle{box}=[shape=rectangle, text height=1.5ex, text depth=0.25ex, yshift=0.2mm, fill=none, draw=black, minimum height=3mm, minimum width=5mm, font={\small}]
\tikzstyle{Z dot}=[inner sep=0mm, minimum size=2mm, shape=circle, draw=black, fill={zx_green}, tikzit fill=green]
\tikzstyle{Z phase dot}=[minimum size=5mm, font={\footnotesize\boldmath}, shape=rectangle, rounded corners=2mm, inner sep=0.2mm, outer sep=-2mm, scale=0.8, tikzit shape=circle, draw=black, fill={zx_green}, tikzit draw=blue, tikzit fill=green]
\tikzstyle{X dot}=[Z dot, shape=circle, draw=black, fill={zx_red}, tikzit fill=red]
\tikzstyle{X phase dot}=[Z phase dot, tikzit shape=circle, tikzit draw=blue, fill={zx_red}, font={\footnotesize\color{black}\boldmath}, tikzit fill=red]
\tikzstyle{H box}=[hadamard]
\tikzstyle{st}=[star, star points=5, fill=white, draw=black, inner sep=1.2pt, line width=1.2pt, tikzit fill=blue, tikzit draw=red, tikzit category=ZH-pf]
\tikzstyle{triangle}=[regular polygon, regular polygon sides=3, fill=white, draw=black, inner sep=0pt, minimum width=1em, tikzit draw=blue, tikzit category=ZH-pf, tikzit fill=cyan]
\tikzstyle{not}=[fill={rgb,255: red,180; green,180; blue,180}, draw=black, shape=circle, font={$\neg$}, dot]
\tikzstyle{vertex}=[inner sep=0mm, minimum size=1mm, shape=circle, draw=black, fill=black]
\tikzstyle{vertex set}=[inner sep=0mm, minimum size=1mm, shape=circle, draw=black, fill=white, font={\footnotesize\boldmath}]
\tikzstyle{wide point}=[fill=white, draw, shape=isosceles triangle, shape border rotate=-90, isosceles triangle stretches=true, inner sep=0pt, minimum width=1.5cm, minimum height=6.12mm, yshift=-0.0mm]
\tikzstyle{medium gray box}=[semilarge box, fill={rgb,255: red,180; green,180; blue,180}]
\tikzstyle{small box}=[rectangle, inline text, fill=white, draw, minimum height=5mm, yshift=-0.5mm, minimum width=5mm, font={\small}]
\tikzstyle{small gray box}=[small box, fill={rgb,255: red,180; green,180; blue,180}]
\tikzstyle{medium box}=[rectangle, inline text, fill=white, draw, minimum height=5mm, yshift=-0.5mm, minimum width=8mm, font={\small}]
\tikzstyle{ddot}=[line width=1.6pt, inner sep=0mm, minimum width=2.5mm, minimum height=2.5mm, draw, shape=circle]
\tikzstyle{dd white}=[ddot, fill=white, tikzit draw=green]
\tikzstyle{dd white phase}=[white phase dot, line width=1.6pt, tikzit draw=yellow]
\tikzstyle{dd gray}=[ddot, fill={rgb,255: red,180; green,180; blue,180}, tikzit draw=green]
\tikzstyle{dd gray phase}=[gray phase dot, line width=1.6pt, tikzit draw=yellow]
\tikzstyle{U box}=[fill=white, draw=black, shape=rectangle, minimum width=1cm, minimum height=1.2cm]
\tikzstyle{tall box}=[fill=none, draw=black, shape=rectangle, minimum width=1.5cm, minimum height=3.5cm]
\tikzstyle{G box}=[fill=white, draw=black, shape=rectangle, minimum width=1cm, minimum height=1.5cm]
\tikzstyle{large G}=[fill={rgb,255: red,11; green,129; blue,255}, draw=black, shape=circle, minimum height=2.5cm]

\tikzstyle{simple}=[-]
\tikzstyle{hadamard edge}=[-, dashed, dash pattern=on 2pt off 1pt, thick, draw=gray]
\tikzstyle{gray}=[-, draw={blue!60!white}, tikzit draw=blue]
\tikzstyle{blue}=[-, draw={blue!60!white}, tikzit draw=blue]
\tikzstyle{brace edge}=[-, tikzit draw=blue, decorate, decoration={brace,amplitude=1mm,raise=-1mm}]
\tikzstyle{diredge}=[->]
\tikzstyle{not edge}=[-, dashed, dash pattern=on 2pt off 1.5pt, thick, draw={rgb,255: red,255; green,68; blue,68}]
\tikzstyle{double edge}=[-, double, shorten <=-1mm, shorten >=-1mm, double distance=2pt]
\tikzstyle{boldedge}=[-, line width=1.6pt, shorten <=-0.17mm, shorten >=-0.17mm, tikzit draw=blue]

\usepackage{float}
\usepackage[T1]{fontenc} %
\usepackage{times} %
\usepackage{color,graphicx} %
\usepackage{subfig,array} %
\usepackage{amsthm,amssymb,amsmath} %
\usepackage[colorlinks=true, linkcolor=blue, citecolor=green, urlcolor=blue, anchorcolor=black]{hyperref}
\usepackage{comment}
\newtheorem{definition}{Definition} %
\newtheorem{lemma}{Lemma} %
\newtheorem{theorem}{Theorem} %
\newtheorem{proposition}{Proposition} %
\newtheorem{example}{Example} %
\newtheorem{procedure}{Procedure}

\newcommand{\nc}{\newcommand} %
\nc{\cA}{{\cal A}} \nc{\cB}{{\cal B}} \nc{\cC}{{\cal C}} %
\nc{\cD}{{\cal D}} \nc{\cE}{{\cal E}} \nc{\cF}{{\cal F}} %
\nc{\cG}{{\cal G}} \nc{\cH}{{\cal H}} \nc{\cI}{{\cal I}} %
\nc{\cJ}{{\cal J}} \nc{\cK}{{\cal K}} \nc{\cL}{{\cal L}} %
\nc{\cM}{{\cal M}} \nc{\cN}{{\cal N}} \nc{\cO}{{\cal O}} %
\nc{\cP}{{\cal P}} \nc{\cQ}{{\cal Q}} \nc{\cR}{{\cal R}} %
\nc{\cS}{{\cal S}} \nc{\cT}{{\cal T}} \nc{\cU}{{\cal U}} %
\nc{\cV}{{\cal V}} \nc{\cW}{{\cal W}} \nc{\cX}{{\cal X}} %
\nc{\cZ}{{\cal Z}}

\begin{document}

%% End-Of-Header

\title{A ZX-Calculus Approach for the Construction of Graph Codes}

\author{Zipeng Wu}
\email{zwubp@connect.ust.hk}
\affiliation{Department of Physics, Hong Kong University of Science and Technology}

\author{Song Cheng}
\email{chengsong@bimsa.cn}
\affiliation{Yanqi Lake Beijing Institute of Mathematical Sciences and Applications}

\author{Bei Zeng}
\affiliation{Department of Physics, Hong Kong University of Science and Technology}%

\date{\today}% It is always \today, today,
             %  but any date may be explicitly specified

\begin{abstract}
Quantum Error-Correcting Codes (QECCs) play a crucial role in enhancing the robustness of quantum computing and communication systems against errors. Within the realm of QECCs, stabilizer codes, and specifically graph codes, stand out for their distinct attributes and promising utility in quantum technologies. This study underscores the significance of devising expansive QECCs and adopts the ZX-calculus—a graphical language adept at quantum computational reasoning—to depict the encoders of graph codes effectively. Through the integration of ZX-calculus with established encoder frameworks, we present a nuanced approach that leverages this graphical representation to facilitate the construction of large-scale QECCs. Our methodology is rigorously applied to examine the intricacies of concatenated graph codes and the development of holographic codes, thus demonstrating the practicality of our graphical approach in addressing complex quantum error correction challenges. This research contributes to the theoretical understanding of quantum error correction and offers practical tools for its application, providing objective advancements in the field of quantum computing.

\end{abstract}

\maketitle

\section{Introduction}

Quantum Error-Correcting Codes (QECCs) play a crucial role in the advancement of
 quantum computing and quantum communication systems. Due to the inherently fragile 
 nature of quantum information, it is highly susceptible to errors arising from environmental 
 noise, imperfect quantum gates, and measurement inaccuracies. These errors can lead to 
 significant losses in computational accuracy and jeopardize the overall performance of 
 quantum systems. To counteract these issues, QECCs have been developed to detect and 
 correct errors without disturbing the delicate quantum states, thus safeguarding the
  integrity of quantum information and enabling the realization of fault-tolerant quantum
   computing~\cite{nielsen2002quantum}.

Stabilizer codes are prominent families of QECCs that have garnered considerable 
attention due to their unique properties and potential applications. Stabilizer codes are a generalization of additive classical codes to the quantum domain. They are characterized
 by a set of stabilizer operators, which preserve the encoded quantum state and provide a 
 framework for efficiently detecting and correcting
  errors~\cite{gottesman1997stabilizer,calderbank1998quantum}. Graph codes, a subclass of 
  stabilizer codes, utilize graph states to represent codewords, where each vertex 
  corresponds to a qubit, and edges represent entangling operations~\cite{briegel2001persistent,raussendorf2001one,hein2006entanglement,schlingemann2001stabilizer,schlingemann2001quantum,grassl2002graphs}. The graph structure provides an intuitive way of visualizing the quantum state and its interactions, making these codes particularly appealing for constructing quantum codes with desirable error-correcting capabilities~\cite{cross2009codeword,Chuang_2009,chen2008nonbinary}.
   Additionally, every graph code is local Clifford equivalent to a stabilizer code,
    which means that they can be transformed into each other through local Clifford operations~\cite{van2004graphical,dehaene2003clifford,hostens2005stabilizer}. The graph code is extensively used in photonic measurement-based quantum communication~\cite{azuma2015all,borregaard2020one} and computation~\cite{bartolucci2023fusion}.
    
In the pursuit of developing future large-scale quantum computers, the capability to design complex QECCs spanning a vast number of qubits becomes increasingly crucial. The intricacies of constructing such extensive QECCs lie in their ability to safeguard quantum information across a broad qubit array while maintaining manageable resource overheads and ensuring high fault tolerance. Concatenated quantum codes emerge as a pivotal strategy in this context, offering exponential enhancements in error-correction efficacy relative to the polynomial growth in resource requirements. These codes adeptly fortify quantum data against errors by layering multiple quantum code strata, underscoring their significance in realizing robust, large-scale quantum computational architectures~\cite{knill1996concatenated,knill1996threshold,knill1998resilient,zalka1996threshold,aharonov1997fault}.Moreover, concatenated codes' adaptability permits the amalgamation of diverse QECC types, such as stabilizer and graph codes, fostering innovative hybrid frameworks that capitalize on each code type's strengths to optimize error correction~\cite{grassl2009generalized}. This versatility is further expanded by recent methodologies like the Quantum Lego framework~\cite{cao2022quantum} and Tensor Network codes~\cite{farrelly2021tensor}, which facilitate the modular construction of expansive quantum error-correcting codes by uniting smaller code segments, thereby extending the traditional concept of code concatenation. A notable application of these advanced frameworks is the creation of holographic quantum codes, where a stabilizer code is devised through the contraction of a tensor network.

The linear mapping of logical to physical qubits, termed the ``encoder'', plays a fundamental role in the construction of larger codes. Employing the graphical ZX-calculus, grounded in category theory and represented through ZX-diagrams, provides an insightful perspective on quantum operations and states, especially within graph codes \cite{Coecke2007graphicalcalculus,DuncanPerdrixGraphStates,van2020zx}. Our research, elaborated in Section~\ref{sec:zxencoder}, harnesses this graphical framework to cohesively link graph codes with their encoding circuits in a ``graph-like'' ZX-calculus format. This novel approach not only delineates the stabilizer generators of graph codes but also broadens the scope to encompass stabilizer codes, enabling the modular construction of more extensive codes. The employment of a simplification technique for Clifford ZX-diagrams, as introduced in \cite{Duncan_2020}, optimizes the encoder for larger codes by amalgamating smaller entities, thereby simplifying the extraction of stabilizers, logical operators, and encoding circuits.

In Section~\ref{sec:Concatenation}, we examine the process of code concatenation, which is crucial for enhancing error correction in quantum systems. Utilizing the ZX encoder diagram, we revisit the generalized local complementation rule from \cite{beigi2011graph} and explore how graph codes equipped with Clifford encoders can be seamlessly concatenated. This analysis extends to concatenated stabilizer codes, demonstrating how the ZX encoder diagram can effectively represent the intricate structure of these codes.

In Section~\ref{sec:fusing}, we demonstrate the use of the ZX encoder diagram to construct larger stabilizer codes from smaller ones, surpassing traditional methods of code concatenation. We apply this technique in our construction of the hyperbolic pentagon (HaPPY) code, as detailed in \cite{pastawski2015holographic}, showing it to be more computationally efficient and graphically intuitive than tensor network contraction and algebraic approaches. By comparing our results with recent studies, such as \cite{munne2022engineering}, we illustrate that our graphically driven method outperforms standard algebraic techniques, making it an invaluable tool for the development of quantum codes.

 In summary, our work harnesses the analytical power of ZX-calculus to offer a new dimension in the understanding and construction of quantum error-correcting codes. By innovatively applying this graphical language, we provide not only a theoretical framework but also practical tools for advancing the field of quantum error correction.

\section{Preliminaries}

In this section, we provide an overview of the preliminary concepts related to graph states, graph codes, and ZX-calculus. These concepts lay the groundwork for the subsequent sections of the paper.

\subsection{\textbf{Graph States, Graph Codes, and Stabilizer Codes}}
\label{sec:graph_code}
This subsection revisits graph codes and explores their intrinsic connection with stabilizer codes, highlighting that every stabilizer code can be locally Clifford equivalent to a graph code.

\textbf{Pauli Group and Stabilizer Code}: 
The foundation of our discussion is stabilizer states, which are closely associated with the Pauli group $\mathcal{P}_n$. The Pauli group is generated by the Pauli matrices $X$ and $Z$.
For simplicity, a Pauli operator $X^{\mathbf{a}} Z^{\mathbf{b}}$ is denoted by the vector $[\mathbf{a} \mid \mathbf{b}]$ of length $2 n$. Thus two Pauli operators $X^{\mathbf{a}} Z^{\mathbf{b}}$ and $X^{\mathbf{a}^{\prime}} Z^{\mathbf{b}^{\prime}}$ commute iff their corresponding vectors are orthogonal with respect to the "symplectic inner product" defined by
\begin{equation}
    [\mathbf{a} \mid \mathbf{b}] *\left[\mathbf{a}^{\prime} \mid \mathbf{b}^{\prime}\right]=\mathbf{a b}^{\prime}-\mathbf{a}^{\prime} \mathbf{b}
\end{equation}

A $[[n,k,d]]$ \textit{stabilizer code} is a type of quantum error-correcting code that encodes $k$ logical qubits into $n$ physical qubits and possesses a code distance $d$. It is defined by an Abelian subgroup $\mathcal{S}$ of $\mathcal{P}_n$, with the negative identity operator $-I$ excluded. For a quantum state $|\psi\rangle$ on $n$ qubits to be a valid codeword of this stabilizer code, it must be stabilized by all elements of $\mathcal{S}$, i.e., $S|\psi\rangle = |\psi\rangle$ for every $S \in \mathcal{S}$. The set of all such stabilized states forms the codespace, a $2^k$-dimensional subspace within the larger $2^n$-dimensional Hilbert space.

The stabilizer group \(\mathcal{S}\) is generated by \(n-k\) commuting operators from \(\mathcal{P}_n\), denoted as \(\mathcal{S} = \langle g_1, \dots, g_{n-k} \rangle\). Consequently, the stabilizer generators can be represented by an \((n-k) \times 2n\) binary check matrix \(C=[A \mid B] \), where the rows correspond to the generators of \(\mathcal{S}\), with matrix \(A\) representing the $X$-component and matrix \(B\) representing the $Z$-component of the generators.

An \(n\)-qubit \textit{stabilizer state} can be viewed as a particular instance of a stabilizer code, specifically a \([[n,0,d]]\) stabilizer code, where the entire \(n\)-qubit space is stabilized by \(n\) linearly independent generators. In this scenario, the codespace is one-dimensional, corresponding to the stabilizer state itself.

\textbf{Clifford group}: The Clifford group \(\mathcal{C}_n\) plays a pivotal role in quantum computing and quantum error correction as it normalizes the Pauli group \(\mathcal{P}_n\).
\begin{equation}
    \mathcal{C}_n := \{U \in \mathcal{U}(2^n) \mid U \mathcal{P}_n U^{\dagger} = \mathcal{P}_n\}
\end{equation}where \(\mathcal{U}(2^n)\) is the group of \(2^n \times 2^n\) unitary matrices. 
The Clifford group can be generated  by three gates: Hadamard, Phase gate and the controlled-X(CX) gate.
%One of the key properties of the Clifford group is its action on stabilizer states.If \(|\psi\rangle\) is a stabilizer state stabilized by the group \(\mathcal{S}\), and \(\mathcal{C}l\) is a Clifford operator, then \(\mathcal{C}l|\psi\rangle\) is also a stabilizer state. The new stabilizer group for \(\mathcal{C}l|\psi\rangle\) can be obtained by conjugating the elements of \(\mathcal{S}\) by \(\mathcal{C}l\), i.e., \(\mathcal{C}l\mathcal{S}\mathcal{C}l^{\dagger}\).

Stabilizer states are said to be locally Clifford equivalent if one can be transformed into the other by applying local Clifford operations, which are operations from the Clifford group acting independently on each qubit. These operations preserve the entanglement structure of the states.

\textbf{Graph State and Graph Code}: A graph state, denoted as \( |G\rangle \), is a specific type of stabilizer state associated with a graph \( G=(V,E) \), where the vertices \( V \) correspond to qubits and the edges \( E \) represent pairwise entanglements created by controlled-Z (\( \mathrm{CZ} \)) operations. To construct \( |G\rangle \), one applies \( \mathrm{CZ} \) gates between qubits for each edge in \( E \), starting from an initial state of all qubits in the superposition state \( |+\rangle^{\otimes |V|}\):
%, where \( n \) is the number of vertices:

\begin{equation}
\label{eq:graph_state} 
    |G\rangle := \prod_{(u, v) \in E} \mathrm{CZ}_{u, v} |+\rangle^{\otimes |V|}
\end{equation}
The stabilizer formalism provides an alternative description of a graph state. Each \( n \)-qubit graph state has a set of stabilizer generators that can be succinctly represented in binary form:

\begin{equation}
    S_G = \left[ I_n \ | \  G \right]
\end{equation}

We also use $G$ to denote  the adjacency matrix of the graph \( G \) for simplicity.

A \( [[n,k,d]] \) \textit{graph code}, defined by a graph \( G \) and a classical code \( \mathcal{C} \), is known to be locally Clifford equivalent to a stabilizer code. The classical codewords \( \mathcal{C} = (\boldsymbol{c}_1, \ldots, \boldsymbol{c}_K) \), with \( K = 2^k \), are instrumental in constructing the codespace, which is spanned by states of the form \( \{Z^{\boldsymbol{c}_1}|G\rangle, \ldots, Z^{\boldsymbol{c}_K}|G\rangle\} \). From \( \mathcal{C} \), we can extract a set of \( k \) linearly independent basis vectors \( \boldsymbol{\alpha} = (\boldsymbol{\alpha}_1, \ldots, \boldsymbol{\alpha}_k) \). These vectors can be organized into a matrix \( \Gamma \) as follows:

\begin{equation}
    \Gamma :=
    \begin{bmatrix}
        \boldsymbol{\alpha_1} \\
        \vdots \\
        \boldsymbol{\alpha_k}
    \end{bmatrix},
\end{equation}
%which efficiently encapsulates the classical codewords within the quantum context.

\textbf{Encoding of Graph Codes}:
To elucidate the encoding process, it is necessary to specify the logical Pauli operators and, consequently, the encoded logical states. 
We denote a graph code with encoded logical states as
\begin{equation}
    |\overline{i_1\ldots i_k} \rangle_G = Z^{i_1\boldsymbol{\alpha}_1 + i_2\boldsymbol{\alpha}_2 + \ldots + i_k\boldsymbol{\alpha}_k} |G\rangle,
\end{equation}
which represents standard form graph codes.
Standard form graph codes are encompassed within the versatile CWS (Codeword Stabilized) quantum code framework \cite{cross2009codeword}, which integrates both additive and nonadditive codes. This unification has facilitated the construction of high-quality quantum codes \cite{Chuang_2009,grassl2009generalized}. The CWS framework's advantage lies in its comprehensive approach, which allows for the application of its analytical tools to improve error correction strategies for graph codes.

The encoding process for a standard form graph code involves both the top $k$ qubits, labeled $q_1, \ldots, q_k$, and the bottom $n$ qubits, labeled $Q_1, \ldots, Q_n$. The step-by-step procedure for encoding such a code is as follows:
\begin{procedure}(\textbf{Encoding of Standard Form Graph Code})
    \label{pro:enc}
    \begin{enumerate}
        \item Prepare the graph state $|G\rangle$ using the circuit in Eq.(\ref{eq:graph_state}) for qubits $Q_1, \ldots, Q_n$.
        \item For each qubit $q_i$, where $1 \leq i \leq k$, and each qubit $Q_j$, where $1 \leq j \leq n$, apply a $\mathrm{CZ}$ gateza if the $j$-th component of vector $\boldsymbol{\alpha}_i$, denoted $\Gamma_{ij}$, is non-zero.
        \item Apply the Hadamard gate $H$ to qubits $q_1, \ldots, q_k$.
        \item Measure qubits $q_1, \ldots, q_k$ in the computational basis and repeat until the outcome $|0\rangle^{\otimes k}$ is obtained.
    \end{enumerate}
\end{procedure}

\begin{figure}[H]
    \centering
    \includegraphics[width=2.5in]{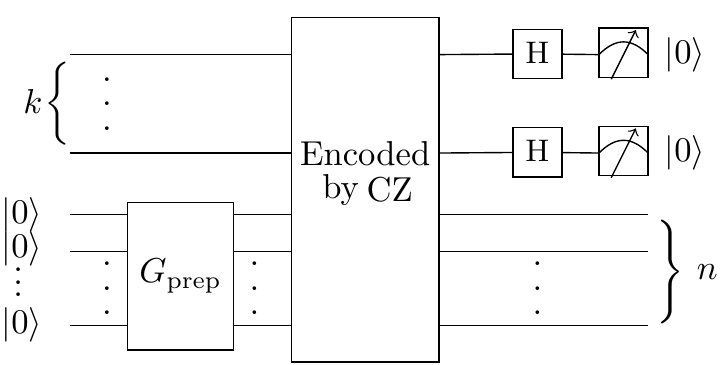}
    \caption{Encoding circuit for a standard form graph code: $k$ qubits are encoded into $n$ qubits. Here, $G_\mathrm{prep}$ represents the encoding circuit for the graph state $|G\rangle$. The "Encoded by CZ" step entails the encoding of the \( k \) qubits into an \( n \)-qubit graph code using CZ gates in accordance with the codeword bases \( \Gamma \).}
    \label{fig:encoder_circuit}
\end{figure}

Starting with the state $\sum_{i_1, \ldots, i_k} a_{i_1, \ldots, i_k}|i_1 \ldots i_k\rangle \otimes |0\rangle^{\otimes n}$, the encoding procedure applies the following transformations:
\begin{eqnarray}
    &&\sum_{i_1, \ldots, i_k} a_{i_1, \ldots, i_k}|i_1 \ldots i_k\rangle \otimes |0\rangle^{\otimes n} \\
    &\mapsto& \sum_{i_1, \ldots, i_k}  a_{i_1, \ldots, i_k}|i_1 \ldots i_k\rangle \otimes |G\rangle\\
    &\mapsto& \sum_{i_1,\ldots,i_k} a_{i_1,\ldots,i_k}|i_1 \ldots i_k\rangle \otimes  Z^{i_1\boldsymbol{\alpha}_1 + \ldots + i_k\boldsymbol{\alpha}_k} |G\rangle\\
    &\mapsto& \frac{1}{\sqrt{2^k}}\sum_{i_1\ldots i_k}\sum_{j_1\ldots j_k} (-1)^{i_1j_1+\cdots+i_kj_k} |j_1\ldots j_k\rangle \otimes a_{i_1\ldots i_k}  Z^{i_1\boldsymbol{\alpha}_1 + \ldots + i_k\boldsymbol{\alpha}_k} |G\rangle
\end{eqnarray}

After measuring qubits \( q_1\ldots q_k \) and obtaining \( |0\rangle^{\otimes k} \), the state of \( Q_1\ldots Q_n \) collapses to \( \sum_{i_1\ldots i_k}a_{i_1\ldots i_k}|\overline{i_1\ldots i_k} \rangle_G \), indicating the mapping \( \sum_{i_1, \ldots, i_k} a_{i_1, \ldots, i_k}|i_1 \ldots i_k\rangle \mapsto  \sum_{i_1, \ldots, i_k} a_{i_1, \ldots, i_k}|\overline{i_1 \ldots i_k}\rangle_G \) is complete.

We will present the stabilizers and logical Pauli operators for the standard form graph code. The choice of basis \(\boldsymbol{\alpha}\) does not alter the codespace, provided the bases span the same space. We select a basis such that \(\Gamma\) assumes a particular form, known as the row-reduced echelon form (RREF), subsequent to the appropriate labeling of the qubits. Consequently, \(\Gamma\) is expressed as \(\Gamma = \left[I_k \ \ M\right]\), where \(M\) is a \(k \times (n-k)\) binary matrix.

The graph \(G\)'s adjacency matrix is split into blocks:

\[
  G=\left[\begin{array}{cc}
    G_1 & N \\
    N^T & G_2
    \end{array}\right],
\]

Here, \(G_1\), \(G_2\), and \(N\) are matrices sized \(k \times k\), \((n-k) \times (n-k)\), and \(k \times (n-k)\) respectively.

Then, the stabilizer generator \(C_G\), and the logical \(\bar{X}\) and \(\bar{Z}\) operators are derived as \cite{beigi2011graph}:

\begin{equation}
\label{eqn:check_matrix}
  \left[\begin{array}{c}
    C_G \\
    \bar{Z} \\
    \bar{X}
    \end{array}\right] = \left[\begin{array}{cc|cc}
    M^T & I_{n-k} & M^TG_1+N^T & M^TN+G_2 \\
    0 & I_k & N^T & G_2 \\
    0 & 0 & I_k & M
    \end{array}\right]    
\end{equation}

The prior discussion has predominantly focused on the standard form of graph codes, wherein the encoded state \(|\overline{i_1\ldots i_k}\rangle_G\) is predefined. Within the realm of error correction, it is imperative to acknowledge that the code distance is inherently dictated by the characteristics of the codespace. Before initiating the encoding process, a basis transformation can be executed, effectively converting the state \(|\overline{i_1\ldots i_k}\rangle_G\) into a novel, complete, and orthonormal state within the codespace.

More specifically, a unitary operation \(U\) can be applied to the top \(k\) qubits prior to the encoding circuit illustrated in Figure \ref{fig:encoder_circuit}. To ensure the logical Pauli operators remain within the Pauli group, \(U\) must be a Clifford operation.

The basis transformation facilitated by \(U\) holds significant relevance in the context of concatenated quantum codes, where multiple layers of error correction are implemented. The choice of basis change can markedly affect the code's overall performance. For instance, the application of a Hadamard transform to the inner code before encoding, as employed in the Shor code\cite{Calderbank_1996}, is beneficial for alternating the correction capabilities for $X$ and $Z$ errors across the various layers of the code.

%This transformation results in a new complete orthonormal state in the codespace being used as the encoded logical state. The basis change \(U\) plays a pivotal role in quantum error-correcting codes (QECC) because it directly affects the encoded logical state, thereby influencing the physical implementation of logical gates. This can lead to a reduction in overhead when mapping logical quantum circuits to their physical counterparts, especially considering the architecture of the quantum computer in question. Moreover, within the framework of concatenated quantum codes—where multiple layers of error correction are applied—the choice of basis change can significantly impact the overall performance. For instance, the application of a Hadamard transform \(H^{\otimes k}\) before encoding, as utilized in the Shor code, facilitates the alternating correction of \(X\) and \(Z\) errors across different levels of the code. This strategy of alternating error correction between layers underscores the critical nature of selecting suitable basis changes for the development of efficient quantum error-correcting codes.

\subsection{ZX-calculus}
In this section, we provide an overview of the ZX-calculus, a graphical language that is uniquely suited for expressing and manipulating quantum circuits and qubit-based linear maps. For an in-depth treatment, readers are referred to Ref.~\cite{van2020zx}. By employing diagrams for its notation, the ZX-calculus becomes an invaluable tool for the transformation, optimization, and simplification of quantum circuits\cite{Duncan_2020,Kissinger_2020}.

The ZX-calculus is a diagrammatic language that consists of wires and spiders. Wires entering the diagram from the left serve as inputs, while wires exiting to the right are outputs. Given two diagrams, we can compose them by connecting the outputs of the first diagram to the inputs of the second, or we can form their tensor product by placing the two diagrams side by side.

The core elements of the ZX-calculus include two fundamental components: $Z$ nodes (or $Z$ spiders) and $X$ nodes (or $X$ spiders), which are interconnected by edges to represent quantum entanglement or linear mappings. The corresponding linear maps of these nodes can be precisely expressed in Dirac notation, as illustrated below.
\begin{figure}[H]
	\centering
	\includegraphics[width=4in]{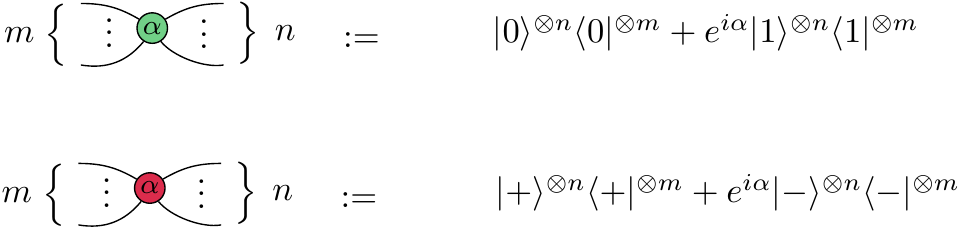}
	%\caption{The Z and $X$ spiders are denoted by green and red nodes, respectively, throughout this paper.}
	%\label{fig:zxspider}
\end{figure}

The Hadamard gate, an essential element in quantum computing, is typically depicted as a white square or a dashed line within the ZX-calculus to simplify the visuals. Below, we also introduce the ZX-diagram representations for the CZ and CX gate.

\begin{figure}[H]
	\centering
	\includegraphics[width=4.3in]{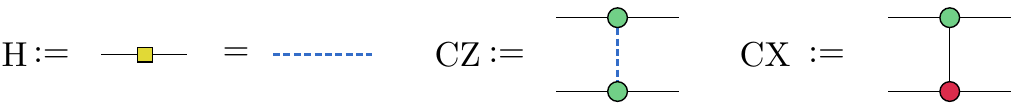}
	%\caption{The Hadamard and controlled-Z gates as represented in ZX-diagrams. The dashed line and white box signify the Hadamard gate.}
	%\label{fig:hadamard_rewrite}
\end{figure}

Furthermore, ZX-diagrams incorporate additional elements such as identity wires, swaps, cups, and caps, further enriching this visual language. Swap operations interpret wire crossings, allowing for the rearrangement of connections. Cups and caps, on the other hand, facilitate the conversion of inputs to outputs and vice versa, enhancing the flexibility and expressiveness of ZX-diagrams, as depicted below:
\begin{figure}[H]
    \centering
    \includegraphics{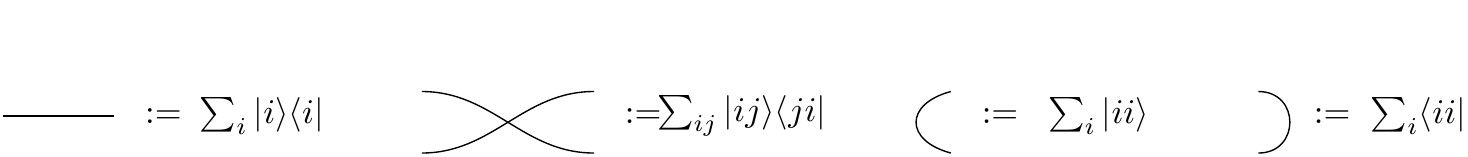}
    %\caption{Illustration of identity wires, swaps, cups, and caps in ZX-diagrams.}
\end{figure}

In the ZX-calculus, the principle that \textbf{``Only connectivity matters''} highlights the framework's elegance and depth. This principle asserts that the specific layout of a ZX-diagram, including the bending and positioning of wires, does not impact the underlying matrix representation, provided that the connections and the sequence of inputs and outputs are preserved. Therefore, any two ZX-diagrams that possess the same number of spiders, identical phase values, and matching connectivity patterns correspond to the same matrix. This concept allows for the flexible interpretation of diagrams as undirected multigraphs, where spiders act as nodes and wires as edges, underscoring the primacy of connectivity.

Our focus is mainly on Clifford ZX-diagrams—specific subsets of ZX-diagrams that represent Clifford circuits. Clifford circuits, known for their ability to transform Pauli operators into other Pauli operators, are integral to quantum computing due to their classical simulability and critical role in quantum error correction.

In Clifford ZX-diagrams, $Z$ and $X$ nodes signify phase shifts by increments of \(\pi/2\). To streamline these diagrams while preserving the represented quantum process, one can employ rewrite rules. Figure ~\ref{fig:clifford_rule} illustrates the \textit{complete} set of rules for simplifying Clifford ZX-diagrams. By sequentially applying these rules, any Clifford ZX-diagram can be transformed into an equivalent, potentially simpler version.

\begin{figure}[H]
	\centering
	\includegraphics[width=4.7in]{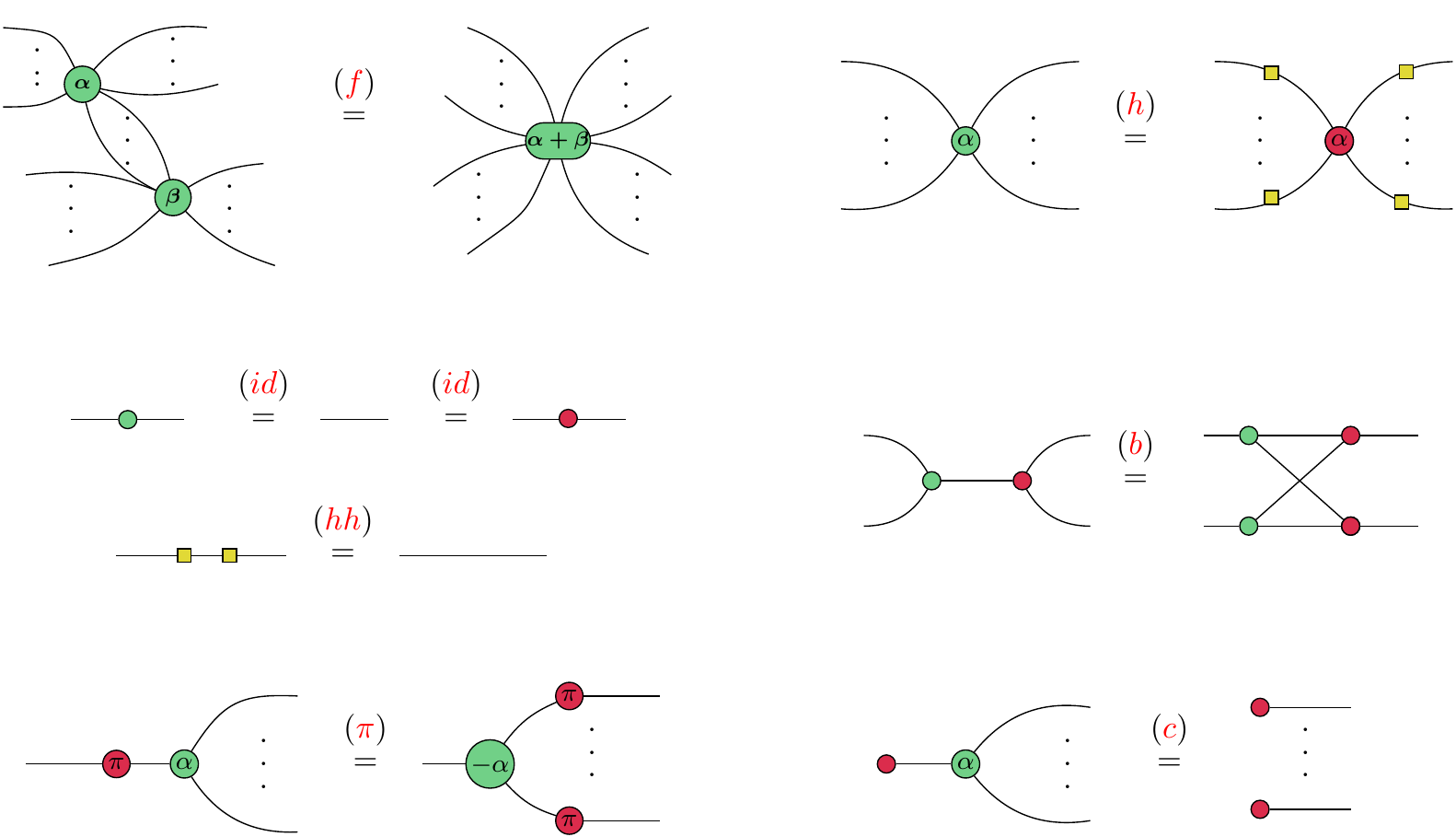}
	\caption{Complete rule sets for Clifford ZX-diagrams. These rules are valid for any \(\alpha, \beta \in [0,2\pi)\). Due to the rules (\(\textcolor{red}{h}\)) and (\(\textcolor{red}{hh}\)), the color interchangeability is also preserved.}
	\label{fig:clifford_rule}
\end{figure}
\textbf{Remark}: In the rules depicted in Figure \ref{fig:clifford_rule}, we deliberately exclude any non-zero scalar factors from consideration. The equality symbol ("=") used in the diagrams denotes equality up to a global non-zero scalar factor. Consequently, within the context of this work, a ZX-diagram is understood to represent a linear map \( L \) up to multiplication by a non-zero scalar.

A key aspect of the ZX-calculus is its illustration of map-state duality, embodied by the Choi-Jamiołkowski isomorphism. This isomorphism enables the conversion of a linear map \( L \) on a qubit system into an equivalent quantum state across \( (n+m) \) qubits. In ZX-diagrams, this conversion is graphically depicted by wire manipulations, notably employing the yanking equation, as shown below:

\begin{figure}[H]
    \centering
    \includegraphics[width=4.5in]{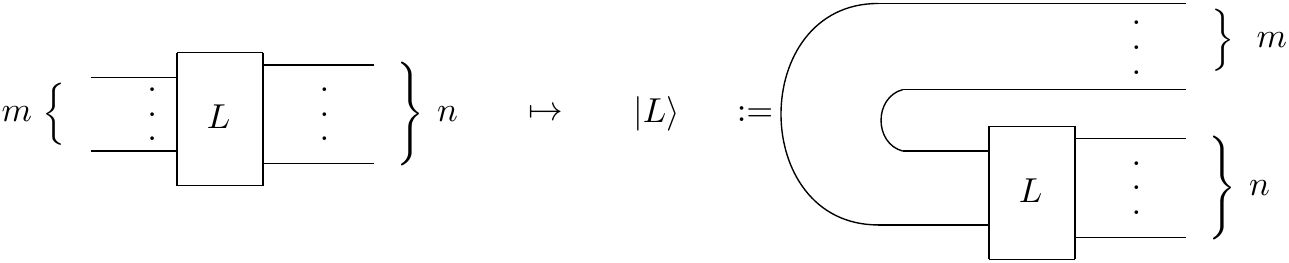}
    %\caption{Illustration of map-state duality in the ZX-calculus.}
    %\label{fig:map_state_duality}
\end{figure}

This graphical transformation is achieved through a partial transpose on the computational basis, creating a one-to-one correspondence between linear transformations \( L: (\mathbb{C}^2)^{\otimes m} \to (\mathbb{C}^2)^{\otimes n} \) and quantum states, denoted \( |L\rangle \), in the space \( (\mathbb{C}^2)^{\otimes(m+n)} \).

Finally, we will discuss the ZX-calculus representation of graph states. The \(|+\rangle\) state is represented as a $Z$ spider with a single output. We then apply the CZ operation to the $Z$ spiders according to the graph. The graph state \(|G\rangle\), associated with graph \(G\), can be represented as follows:

\begin{figure}[H]
	\centering
	\includegraphics[width=1.5in]{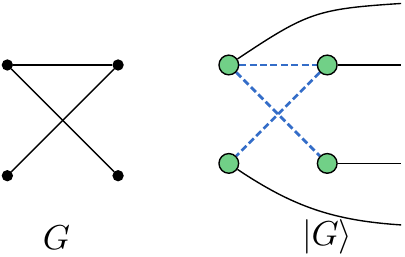}
	%\caption{ZX-diagram translation of a graph into a graph state.}
\end{figure}

\subsection{Simplification of Graph-like ZX-diagrams}
Before delving into the analysis of graph codes, it is essential to introduce the concept of graph-like ZX-diagrams and the associated simplification strategy, as outlined in \cite{Duncan_2020}. This section reviews key concepts and results from \cite{Duncan_2020}, providing a concise overview of the simplification strategy.

The graph-like ZX-diagram is formally defined as follows:

\begin{definition}[Graph-like ZX-diagram]
   A ZX-diagram is considered graph-like if it meets the following criteria:
   \begin{enumerate}
    \item All spiders within the diagram are Z spiders.
    \item Connections between Z spiders are exclusively through Hadamard edges.
    \item The diagram does not contain parallel Hadamard edges or self-loops.
    \item Each input or output is connected to exactly one Z spider, and every Z spider is connected to at most one input or output.
   \end{enumerate}
\end{definition}

\begin{lemma}\label{lema:graph}
  Every ZX-diagram can be transformed into an equivalent graph-like ZX-diagram. 
\end{lemma}

The proof of Lemma~\ref{lema:graph} is outlined in the publication by Duncan et al.~\cite{Duncan_2020}. To prove the lemma, two additional derivable rules are required, which are delineated  below:

\begin{figure}[ht]
	\centering
	\includegraphics[width=4in]{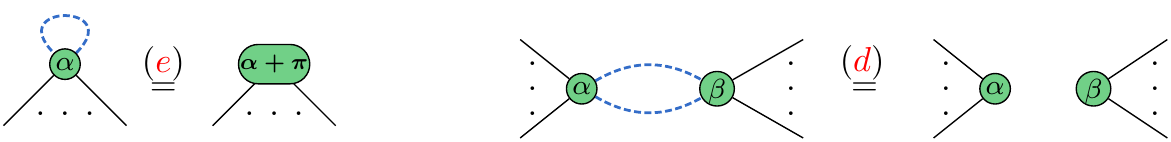}
	%\caption{Additional derivable rules frequently utilized.}
	\label{fig:zxrule}
\end{figure}

The transformation process begins by converting all $X$ spiders into $Z$ spiders, utilizing the (\(\color{red}{h}\)) rule. Subsequently, unnecessary Hadamard edges are eliminated by applying the (\(\color{red}{hh}\)) rule. The next step involves fusing $Z$ spiders as much as possible, leading to a diagram composed of $Z$ spiders interconnected by Hadamard edges. To remove self-loop Hadamards, the (\(\color{red}{e}\)) rule is applied, and parallel Hadamards are managed with the (\(\color{red}{d}\)) rule, ensuring adherence to conditions 1-3 of the graph-like definition. To fulfill condition 4, the (\(\color{red}{id}\)) rule can be used to insert dummy $Z$ spiders.

Lemma \ref{lema:graph} gaurantee to transformed every ZX-diagram into equivalent graph-like ZX-diagram. 
The simplification process for graph-like ZX-diagrams generally involves applying rewrite rules based on local complementation~\cite{bouchet1988graphic} and pivoting~\cite{kotzig1968eulerian} techniques to eliminate as many interior spiders as possible. A spider is called interior when it is not connected to an input or an output, otherwise it is called a boundary spider.
 The detailed derivations of these simplification rules, utilizing ZX-calculus, are comprehensively explained in \cite{Duncan_2020}. 
 
 The simplification rule ($\color{red}lc_1$) resulting from local complementation is illustrated in the figure below:
\begin{figure}[H]
	\centering
	\includegraphics[width=4in]{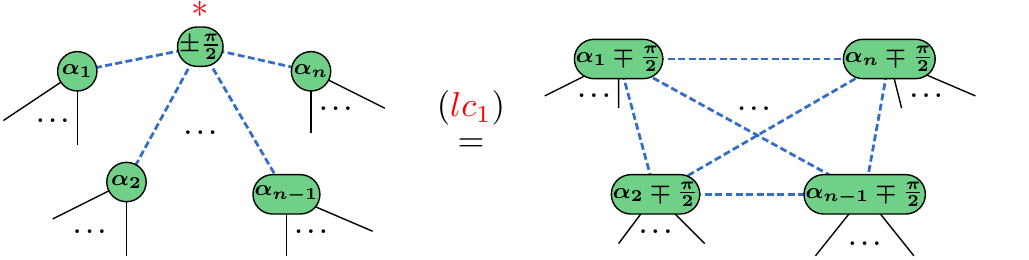}
\caption{Lemma 5.2, referenced in \cite{Duncan_2020}, is based on the application of the local complementation rule. This operation is labeled as the ($\color{red}lc_1$) rule and is executed on the node identified by the $\color{red}*$ symbol.}
  \label{fig:lc}
\end{figure}

The ($\color{red}pv_1$) rule resulting from pivoting is illustrated as:

\begin{figure}[H]
	\centering
	\includegraphics[width=4in]{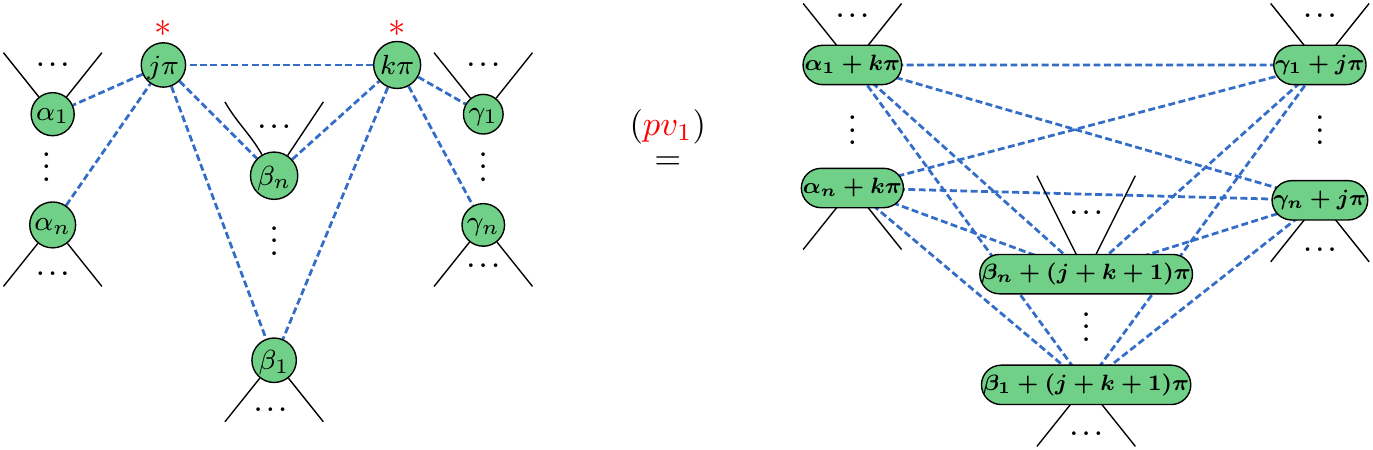}
\caption{Lemma 5.4, referenced in \cite{Duncan_2020},  is based on the application of the pivoting rule.This operation is labeled as the ($\color{red}pv_1$) rule and is executed on the two nodes identified by the $\color{red}*$ symbols.}
  \label{fig:pivot}
\end{figure}
Another generalization of the rule ($\color{red}pv_1$) is the following:
\begin{figure}[H]
	\centering
	\includegraphics[width=6in]{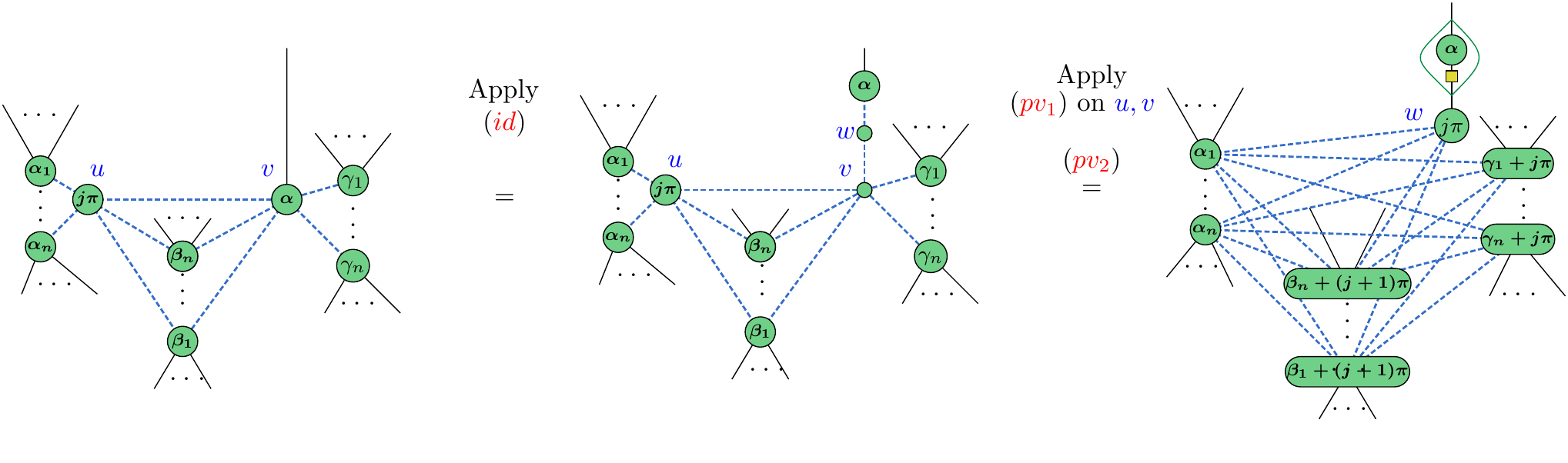}
	\caption{The rule ($\color{red}pv_2$) is a variation of ($\color{red}pv_1$), where rule ($\color{red}id$) is first applied to node $v$, followed by pivoting on nodes $u$ and $v$. This rule facilitates the removal of the interior white node $u$, treating $w$ as a boundary spider, and separately saving the single qubit unitaries (marked by green circles) to the top of $w$.}
  \label{fig:pivot2}
\end{figure}
%These three rules are effective in eliminating interior nodes from graph-like ZX-diagrams. This technique has been proven to significantly simplify quantum circuits, as detailed in \cite{Duncan_2020}.

\begin{theorem}[Simplification of Graph-like ZX-diagrams]
  \label{th:simplification}
  Given any graph-like ZX-diagram \(D\), a terminating procedure exists to transform \(D\) into a graph-like ZX-diagram \(D'\) (up to single-qubit unitaries on inputs/outputs). The resulting diagram \(D'\) will not contain:
  \begin{enumerate}
      \item Interior proper Clifford spiders,
      \item Adjacent pairs of interior Pauli spiders,
      \item Interior Pauli spiders adjacent to a boundary spider.
  \end{enumerate}
  Furthermore, if \(D\) consists solely of Clifford spiders, then \(D'\) will be devoid of interior spiders.
\end{theorem}

This theorem is proven by iteratively applying the rules ($\color{red}lc_1$), ($\color{red}pv_1$), and ($\color{red}pv_2$) to the graph-like ZX-diagram. As we deal with Clifford ZX-diagrams, it's noteworthy to re-emphasize the result of Theorem \ref{th:simplification}: Any graph-like Clifford ZX-diagram can be simplified to a graph-like ZX-diagram with no interior spiders (up to single-qubit unitaries on inputs/outputs). This result is instrumental for the analysis of graph codes, as will be demonstrated in the subsequent section.

\section{ZX Encoder Diagram of Graph Codes and Stabilizer Codes}
\label{sec:zxencoder}
This section explores the conversion of encoding circuits into ZX-diagrams, in alignment with Procedure \ref{pro:enc} and the encoding circuit depicted in FIG.~\ref{fig:encoder_circuit}.

The ZX-calculus is recognized as a versatile tool for representing linear maps from \( L: (\mathbb{C}^2)^{\otimes k} \to (\mathbb{C}^2)^{\otimes n} \), making it particularly suitable for depicting the encoding map of graph codes or, more generally, any binary code. The representation of the encoding map within the ZX-calculus framework is termed the ZX encoder diagram (or simply, the encoder). While the ZX-calculus can depict the encoding map, it does not inherently guarantee the map's simplicity or efficiency. Nonetheless, it will be shown that a graph-like structure within the ZX-calculus provides an optimal representation for graph codes.

Abstractly, the encoder for an \( [[n,k,d]] \) code can be visualized as follows:

\begin{figure}[H]
    \centering
    \includegraphics[width=1.5in]{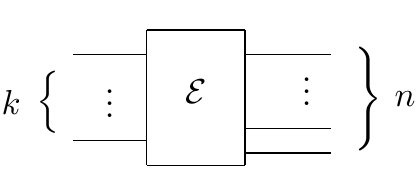}
    \caption{Abstract representation of the ZX encoder diagram for an \( [[n,k,d]] \) code.}
    \label{fig:encoder_abstract}
\end{figure}

Inputting a logical state from \( \mathcal{H}_2^{\otimes k} \) results in an encoded logical state within the codespace, which is a subspace of \( \mathcal{H}_2^{\otimes n} \). 

The encoding process of a standard form graph code can be systematically depicted in a ZX-diagram, starting with the transformation illustrated below:

\begin{figure}[H]
  \centering
  \includegraphics[width=5.4in]{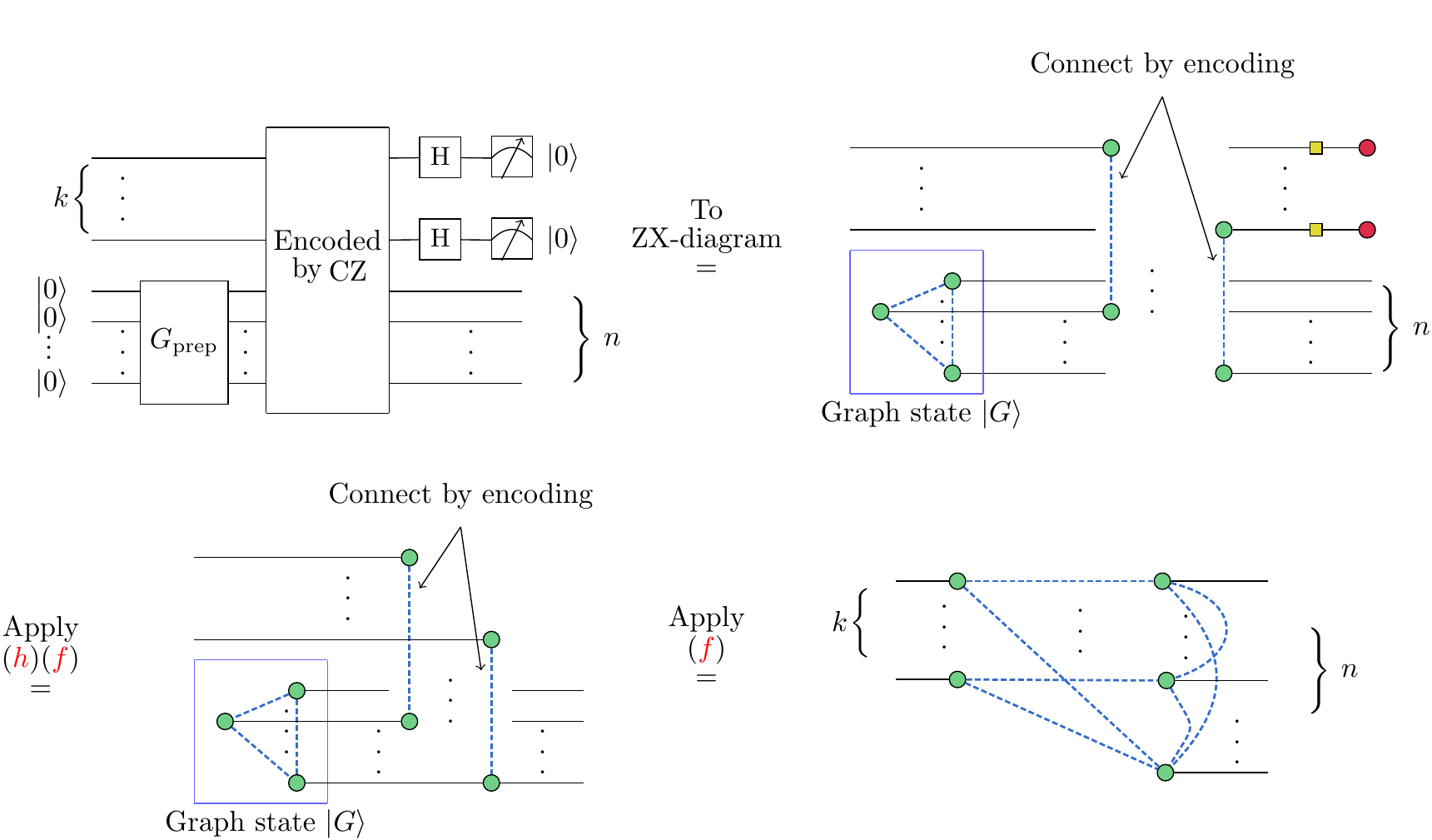}
  \caption{Transformation of the encoding circuit for a standard form graph code into a ZX-diagram.}
  \label{fig:zxmapping}
\end{figure}

This representation of the encoder for a standard form graph code distinctly embodies two pivotal elements of the graph code: the classical codewords \(\Gamma\) and the graph state \(|G\rangle\). The diagram comprises \(k\) $Z$ spiders on the left, symbolizing the input nodes, and \(n\) $Z$ spiders on the right, representing the output nodes. Each input node is linked to a single input edge, and similarly, each output node is connected to only one output edge. The input and output spiders are interconnected via Hadamard edges, with the bi-adjacency matrix between them reflecting the classical codeword \(\Gamma\). The \(n\) output nodes form a cluster interconnected by Hadamard edges in accordance with the graph \(G\), thereby specifying the stabilizer generator through the diagram's connectivity.

Expressed as a linear map, the diagram aligns precisely with the encoding circuit:

\begin{equation}
    \mathcal{E}_G = \sum_{i_1\ldots i_k}|\overline{i_1\ldots i_k}\rangle_G  \langle i_1\ldots i_k|
\end{equation}

As previously indicated in Section \ref{sec:graph_code}, a basis change with a unitary \( U \in \mathcal{U}(2^k) \) for the logical state \( |i_1\ldots i_k\rangle \) can be applied before encoding to tailor the encoding process for specific requirements. This modification does not impact the codespace, and, as a result, does not alter the stabilizer \( \mathcal{S} \) of the code. Diagrammatically, the encoder incorporating a basis change \( U \) is represented as follows:

\begin{figure}[H]
  \centering
  \includegraphics[width=2.1in]{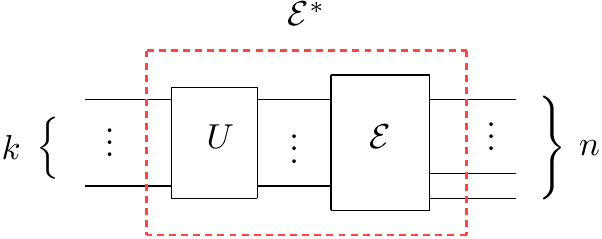}
  \caption{Representation of the encoder with a basis change \( U \).}
  \label{fig:encoder_with_basis_change}
\end{figure}

The encoding map with the basis change is denoted by \( \mathcal{E}^* = \mathcal{E}U \). While this basis change does not affect the codespace and thus preserves the code distance, it serves an important purpose. Specifically, the basis change \( U \) alters the encoded logical state, which in turn affects the physical implementation of logical gates. For instance, in the context of implementing an error-correcting quantum circuit, it is essential to translate the logical quantum circuit into a physical quantum circuit. Considering that the basic gate set available on a quantum computer typically depends on its architecture, an appropriately chosen \( U \) can significantly reduce the overhead of the physical quantum circuit if it facilitates a low-overhead or transversal implementation of the basic gate set in accordance with the architecture's constraints.

Furthermore, in the context of quantum code concatenation discussed in this work, the choice of basis change significantly influences the performance of concatenated codes. A prevalent strategy involves applying a basis change \(U' = H^{\otimes k'}\) (where the prime denotes the inner code) before encoding, thereby transitioning the encoding from the $Z$ basis to the $X$ basis. This technique is notably utilized in the Shor code\cite{Calderbank_1996}, which employs concatenation of the [[3,1,1]] repetition code with \(U' = H\). The objective of this approach is to alternate the correction of $X$ and $Z$ errors across different levels of the code, thereby enhancing the efficacy of error correction.

This discussion can be extended to the encoder of a stabilizer code:

\begin{proposition}
\label{pro:encoder_stb} 
The encoder \(\mathcal{E}\) of a stabilizer code can be depicted as a ZX-diagram that is locally Clifford equivalent to a graph state after applying map-state duality. Diagrammatically, the encoder \(\mathcal{E}\) is represented as follows:

\begin{figure}[H]
    \centering
    \includegraphics[width=2.5in]{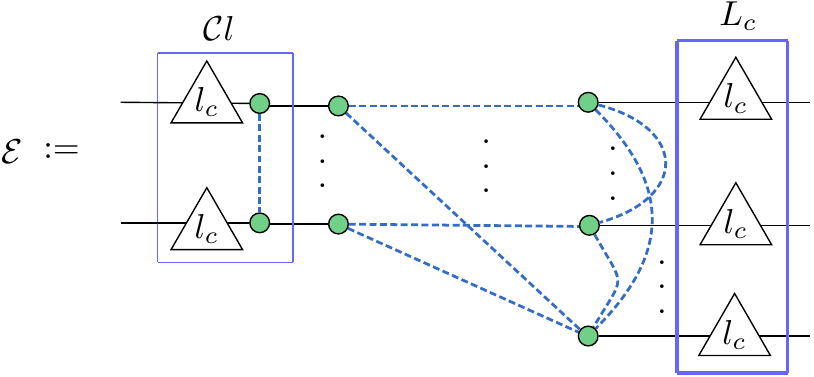}
\caption{ZX-diagram representing the stabilizer code encoder. Here, \(\mathcal{C}l\) denotes a Clifford gate, and \(L_c\) represents a local Clifford gate. The triangles indicate various single-qubit Clifford gates, which are not distinguished in this diagram.}
    \label{fig:encoder_stb}
\end{figure}

\end{proposition}

\begin{proof}
Consider an \([[n,k,d]]\) stabilizer code equipped with a stabilizer \(\mathcal{S} = \langle g_1, \ldots, g_{n-k} \rangle\), where the \(g_i\) represent \(n-k\) independent generators. The encoded logical states are denoted by \(|\overline{i_1 \ldots i_k}\rangle\).

The encoding operation for this stabilizer code is specified by:
\begin{equation}
    \mathcal{E} = \sum_{i_1 \ldots i_k}|\overline{i_1 \ldots i_k}\rangle \langle i_1 \ldots i_k|,
\end{equation}
which, utilizing map-state duality, can be reformulated as:
\begin{equation}
    |\mathcal{E}\rangle = \sum_{i_1 \ldots i_k} |i_1 \ldots i_k\rangle|\overline{i_1 \ldots i_k}\rangle .
\end{equation}
The (unnormalized) state \(|\mathcal{E}\rangle\) is stabilized by:
\begin{equation}
    \langle I^{\otimes k} \otimes g_1, \ldots, I^{\otimes k} \otimes g_{n-k}, X_1 \otimes \overline{X}_1, Z_1 \otimes \overline{Z}_1, \ldots, X_k \otimes \overline{X}_k, Z_k \otimes \overline{Z}_k \rangle,
\end{equation}
comprising \(n+k\) independent stabilizer generators. The logical Pauli operators of the stabilizer code are typically selected from the Pauli group \cite{gottesman1997stabilizer}, which confirms that \(|\mathcal{E}\rangle\) qualifies as a stabilizer state. This state is locally Clifford equivalent to a graph state across \(n+k\) qubits \cite{van2004graphical}, allowing for its representation within a ZX-diagram. The transition to the ZX encoder diagram for the stabilizer code is finalized by converting the kets to bras.
\end{proof}

%This proposition and its proof elucidate how the encoder of a stabilizer code can be effectively represented within the ZX-calculus framework, leveraging the local Clifford equivalence to graph states for intuitive and concise diagrammatic depiction.

The evaluation of the ZX encoder diagram from dual perspectives yields valuable insights. Firstly, the elimination of Clifford \(\mathcal{C}l\) and local Clifford \(L_c\) from the diagram converts it into a standard form graph code while preserving the original code distance of the stabilizer code. This preservation is attributable to \(\mathcal{C}l\), which, although it modifies the encoded basis, does not alter the codespace. Similarly, \(L_c\), consisting of local Clifford operations, does not influence the code distance.

Secondly, when constructing large stabilizer codes from smaller ones by stacking the encoder in a modular fashion, such as with concatenated stabilizer codes and holographic codes, the diagrammatic approach facilitates the analysis of more complex coding structures. The encoder for a larger code can be derived by executing the simplification process, as detailed in Theorem~\ref{th:simplification}. Once the fully simplified encoder is obtained, the framework described herein becomes readily applicable to analyzing the larger code. Given the encoder diagram, the stabilizer, logical Pauli operators, and the encoding circuit can be efficiently constructed, as detailed in Appendices \ref{ap:encoder_circuit} and \ref{ap:stab}.

It is important to note that the ZX encoder diagram, as defined, is not inherently unique. A graphically distinct but mathematically equivalent encoder can be generated through any sequence of pivoting and local complementation operations. In situations where only the stabilizer of a stabilizer code is known and the logical operators are unspecified, Khesin et al. \cite{khesin2023graphical} have proposed a canonical form for the encoder and provided a methodology for transitioning to this canonical form.

\begin{proposition}
\label{pro:enc_valid}
  An encoder, as described in Proposition \ref{pro:encoder_stb}, is a valid stabilizer encoder if and only if the \( k \times n \) bi-adjacency matrix \(\Gamma\) between input nodes and output nodes has rank \( k \), where \( n > k \).
\end{proposition}

\begin{proof}
The sufficiency of this condition is evident when considering the encoder as \(\mathcal{E} = L_c \mathcal{E}_G \mathcal{C}l\). If \(\Gamma\) has rank \(k\), then \(\mathcal{E}_G\) serves as a valid graph code encoder. As a result, \(\mathcal{E}\) is also valid, since it essentially performs a basis change on the logical state \(|i_1 \ldots i_k\rangle\) and applies a transformation to the codespace \(C_s\) via \(L_cC_sL_c^\dagger\).

The necessity of this condition stems from the intrinsic structure of the codespace. For the encoding to be effective, the bi-adjacency matrix \(\Gamma\) must have full rank to guarantee that each input state is encoded into a distinct codeword. If \(\Gamma\) lacks full rank, then there exist two distinct binary vectors \(\boldsymbol{x}\) and \(\boldsymbol{y}\) such that their respective linear combinations with the rows of the matrix are identical, i.e., \(x_0\boldsymbol{\alpha}_1 + \ldots + x_k\boldsymbol{\alpha}_k = y_0\boldsymbol{\alpha}_1 + \ldots + y_k\boldsymbol{\alpha}_k\). Consequently, the input states \(|x_0 \ldots x_k\rangle\) and \(|y_0 \ldots y_k\rangle\) would be mapped to the same encoded state, contravening the essential principle of stabilizer codes: each unique input state must correspond to a unique encoded state.
\end{proof}

This proposition becomes helpful when we are building a small stabilizer code encoder into a larger one since it can be used to determine whether an encoder is valid.
To conclude this subsection, we will present an example that illustrates the \( [[7,1,3]] \) Steane code. Beginning with its stabilizer, we can systematically construct the logical Pauli operators and the encoding circuit, as detailed in \cite{gottesman1997stabilizer,grassl2011variations}. The encoding circuit is then converted into a Clifford ZX-diagram, which is subsequently fully simplified to yield the ZX encoder diagram.

\begin{figure}[H]
  \centering
  \includegraphics[width=6.3in]{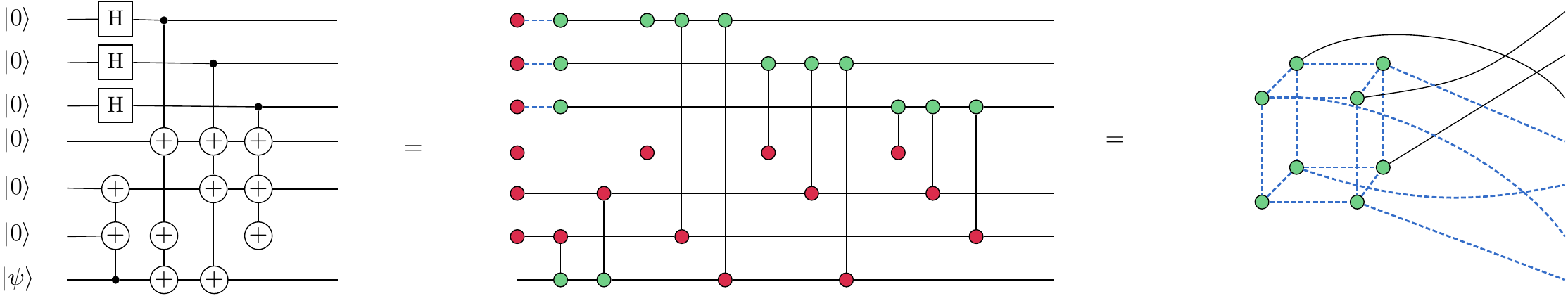}
  \caption{Converting the encoding circuit of \( [[7,1,3]] \) Steane code into a ZX encoder diagram. The Steane code is locally Clifford equivalent to a graph code through a partial Hadamard operation.}
  \label{fig:[7,1,3]}
\end{figure}

With the ZX encoder diagram established, we can concatenate graph codes (and more generally, stabilizer codes) and derive the graphical representation of the concatenated code. Once we perform the simplification and derive the ZX encoder diagram of the concatenated code, the entire code is specified; that is, the stabilizer group and logical Pauli operations. This purely diagrammatic representation is extremely useful for the analysis of concatenated codes and offers advantages over the tableau analysis of stabilizer concatenation.

\textbf{Remark:} This work primarily employs graph-like ZX-diagrams to depict the encoders of stabilizer codes, extending the foundational principles of graph codes. The underlying motivation for this approach is anchored in the simplification capabilities of Clifford ZX-diagrams, as articulated in Theorem~\ref{th:simplification}. Although this methodology proves effective for analyzing stabilizer codes, it may not invariably be the optimal strategy for all code types. For example, in the case of codes with dual structures, such as CSS codes, a phase-free ZX-diagram might present a more fitting representation, as suggested by \cite{kissinger2022phase}. However, considering that every stabilizer code is locally Clifford equivalent to some graph code, the investigation of graph-like encoders continues to be a worthwhile pursuit.

\section{Concatenated Codes}
\label{sec:Concatenation}

Concatenated quantum error-correcting codes are essential for protecting quantum information against decoherence and operational errors. These codes utilize a multi-tiered framework of quantum error correction, which is key to mitigating error propagation across qubits. The structure, featuring nested layers of quantum error correction, is critical for the scalable and fault-tolerant architectures required in quantum computing and communication.

In this section, we delve into the domain of concatenated graph codes, beginning our exploration with an introduction to the fundamentals of concatenated codes. We proceed to leverage the ZX encoder diagram as a vital tool for the examination of graph code concatenation. Our analysis includes re-proving the Generalized Local Complementation Rule, specifically tailored to concatenated graph codes.

The methodology showcased here illustrates the ZX-calculus's capability to manage arbitrary concatenations of graph codes under various bases, that is, for any Clifford operation \( U' \in \mathcal{C}_n \). Moreover, the insights obtained from our analysis of graph codes can be seamlessly extended to the graphical representations of concatenated stabilizer codes. This adaptability emphasizes the versatility and efficacy of the ZX encoder diagram, highlighting its substantial utility in the representation and scrutiny of complex concatenated structures within the field of quantum error correction.

%\subsection{Concatenated Codes}
%\label{sec:Concatenation}

%Concatenated quantum error correction codes are crucial for safeguarding quantum information against decoherence and operational errors, employing a layered approach to quantum error correction that is essential for preventing error spread among qubits. This multi-tiered structure is vital for achieving the scalable, fault-tolerant designs needed in quantum computing and communication.

%The section starts with an overview of concatenation concepts, followed by an analysis of graph code concatenation using the ZX encoder diagram. This includes a re-evaluation of the Generalized Local Complementation Rule for such codes. Furthermore, we show the ZX-calculus's ability to manage various graph code concatenations, including those for any Clifford $U$. The findings from this graph code analysis are applicable to the graphical representation of concatenated stabilizer codes, underscoring the ZX encoder diagram's utility in depicting and examining complex concatenated quantum error correction structures.

\textbf{Concatenated Codes}: A concatenated quantum code combines multiple quantum codes in a hierarchical manner to improve error correction. It consists of an outer quantum code encoding \( k \) logical qubits into \( n \) physical qubits with a minimum distance \( d \), and an inner quantum code encoding \( k' \) logical qubits into \( n' \) physical qubits with a minimum distance \( d' \). Each physical qubit of the outer code is encoded using the inner code, resulting in a concatenated code that encodes \( k \times k' \) logical qubits into \( n \times n' \) physical qubits, with a minimum distance at least equal to \( d \times d' \). 

In the construction of concatenated quantum codes, \( k \times k' \) logical qubits are first partitioned into \( k \) subsets, each comprising \( k' \) qubits. Each subset is encoded into \( n' \) qubits via the inner quantum error-correcting code. These encoded subsets are then treated as new logical units and further encoded using an outer code, expanding each to \( n \) physical qubits. The resultant encoding yields a robust \( n \times n' \) physical qubit array, with each original logical qubit doubly protected. This hierarchical encoding scheme can be recursively applied, creating multiple layers of error correction to improve the quantum computer's resilience against errors\cite{knill1996concatenated}.

\begin{proposition}
  In ZX encoder diagrams, the concatenation of a quantum code is represented by connecting the output of the outer code to the input of the inner code.
\end{proposition}

The encoding circuit for concatenated codes can be systematically converted to a ZX-diagram by regarding the encoding circuit as a ZX encoder diagram.

\begin{figure}[H]
  \centering
  \includegraphics[width=5.5in]{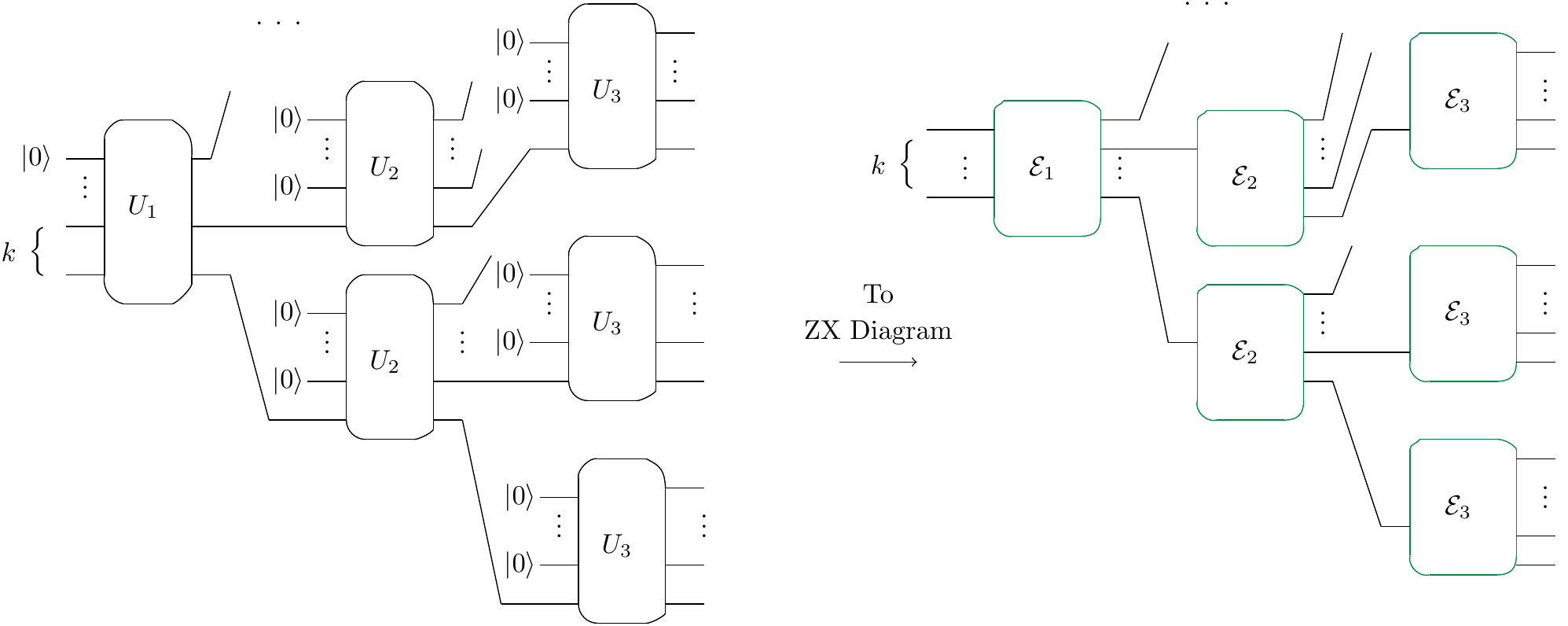}
  \caption{This diagram converts the encoding circuit for two-level concatenated codes into a ZX encoder diagram. For simplicity, only a portion of the qubits being concatenated to the inner code is depicted.}
  \label{fig:concatenate_zx}
\end{figure}

In ZX encoder diagrams, concatenating codes is efficiently realized by stacking encoders, which transform the outer code's output nodes and the inner code's input nodes into interior nodes within the graph-like ZX-diagram. For stabilizer codes, simplification rules applied to these diagrams allow for the removal of these interior nodes, yielding a streamlined representation of the concatenated code. This simplified ZX encoder diagram defines the concatenated code, facilitating the clear tracing of information flow and the derivation of logical operations. The approach simplifies the construction of complex concatenated stabilizer codes and enhances the development of new error correction strategies. Thus, the ZX encoder diagram emerges as a vital tool in quantum error correction, offering both clarity and utility in navigating the complexities of quantum code analysis.

\subsection{Concatenated Graph Codes}
The landscape of concatenated graph codes research includes various results, one of which addresses the graphical representation of concatenated graph codes through the addition of auxiliary vertices. Hein et al.~\cite{hein2006entanglement} explored this notion, but the specific structure of the resultant graph without these auxiliary vertices remains an open question. Generalized concatenated codes have been graphically characterized by Grassl et al.~\cite{grassl2009generalized}, but only for cases where the outer code adheres to a particular structure.

Building upon these foundations, Beigi et al.~\cite{beigi2011graph} developed a systematic method for constructing concatenated graph codes, specifically when the inner code utilizes a basis change unitary of the form \( U' = H^{\otimes k'} \). Their technique employs "generalized local complementation" (GLC) to merge the graphs of the inner and outer codes, establishing the graph of the concatenated code. Despite its utility, the GLC-based method has its constraints. It is not universally applicable, particularly when GLC is not viable, resulting in stabilizer codes that diverge from the typical graph code paradigm. %These constraints are due to the method's limitation to particular unitary operations and its narrow scope in the broader field of stabilizer code concatenation.

To address these limitations, this discussion introduces the ZX encoder diagram as an analytical tool for concatenated codes. The ZX encoder diagram enhances the examination of concatenated codes beyond the scope of Beigi et al.'s method, offering a more versatile approach to proofs and applications.

The subsequent sections will revisit the contributions of Beigi et al.~\cite{beigi2011graph} and demonstrate how the ZX encoder diagram can be effectively applied to the study of concatenated graph codes. Special attention will be given to GLC and its role in this context.

\begin{definition}\textbf{Generalized Local Complementation (GLC)}
Consider a graph \( G \) consisting of two clusters of nodes, denoted as \( N_1 \) and \( N_2 \). The GLC on \( G \) results in a new graph \( G' \), characterized by the following properties:
\begin{itemize}
    \item The subgraphs induced by \( N_1 \) and \( N_2 \) in \( G \) are preserved in \( G' \), meaning the internal structures of these clusters remain unchanged.
    \item The inter-cluster edges between nodes in \( N_1 \) and \( N_2 \) in \( G' \) are the complements of those in \( G \). Specifically, an edge exists between nodes in different clusters in \( G' \) if and only if no corresponding edge exists in \( G \).
\end{itemize}
In essence, GLC maintains the individual cluster structures within \( G \) while inverting the connectivity between them. Figure \ref{fig:glc_def} illustrates the GLC operation applied to clusters \( N_1 \) and \( N_2 \) in graph \( G \).
\end{definition}

\begin{figure}[H]
\centering
\includegraphics[width=3.8in]{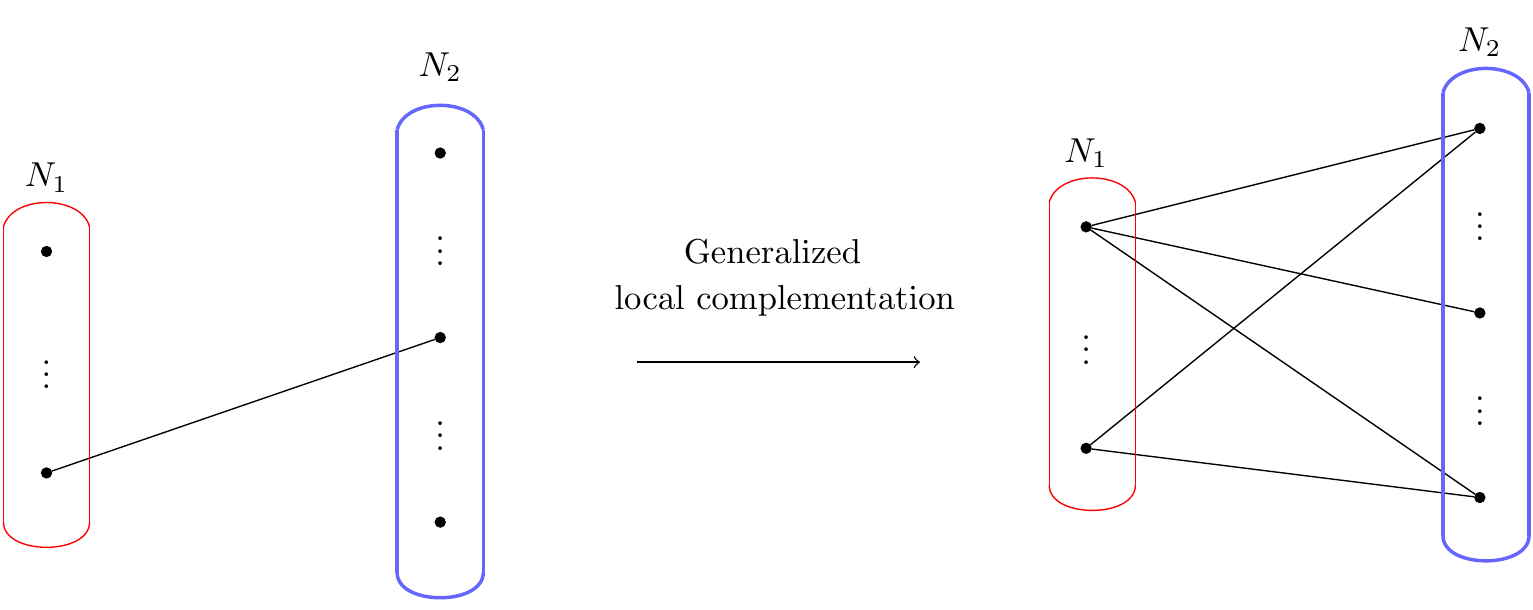}
\caption{Generalized local complementation between clusters \( N_1 \) and \( N_2 \) in graph \( G \).}
\label{fig:glc_def}
\end{figure}

The term ``generalized'' differentiates this operation from the standard local complementation, which typically involves a single node \( i \) and its neighborhood \( N(i) \). In the context of graph code concatenation, as discussed in \cite{beigi2011graph}, the graph of the concatenated code is obtained by sequentially applying the GLC rule between each pair of nodes from the inner and outer codes. Thus, GLC acts as a pivotal operation in creating complex graph structures from simpler components.

With the concept of GLC defined, we now introduce the result obtained by Beigi et al.~\cite{beigi2011graph}.

\begin{theorem}[GLC for Graph Codes]
\label{thm:glc}
Let \( C \) and \( C' \) be outer and inner graph codes in standard form, respectively. Assume \( C' \) undergoes a basis change \( U' = H^{\otimes k'} \). To construct the ZX encoder diagram for the concatenated code, one applies the GLC rule to each pair of Hadamard-connected nodes \( (i, i') \), where \( i \) belongs to \( C \) and \( i' \) to \( C' \). Specifically, the GLC rule is applied to the sets \( N(i) \setminus \{i'\} \) and \( N(i') \setminus \{i\} \), followed by the removal of the nodes \( (i, i') \). This procedure is repeated for all corresponding node pairs in the graph.
\end{theorem}

As an example, consider the self-concatenation of a \([[7,1,3]]\) graph code. The resulting concatenated code is a \([[49, 1, 9]]\) graph code. Figure \ref{fig:713_concate} offers a visual representation of this process:

\begin{figure}[H]
    \centering
    \includegraphics[width=5in]{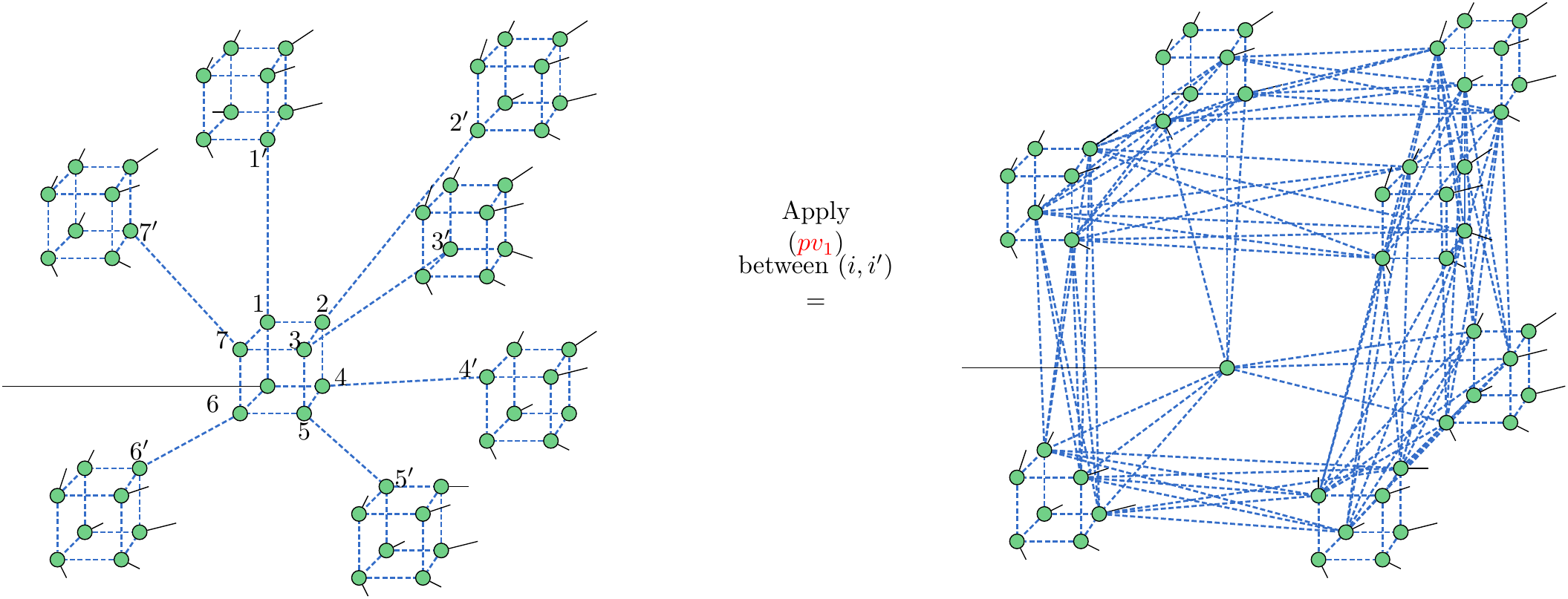}
    \caption{Self-concatenation of the \([[7, 1, 3]]\) graph code, resulting in a \([[49, 1, 9]]\) code. The step-by-step derivation using ZX-calculus is provided in Appendix \ref{ap:713}.}
    \label{fig:713_concate}
\end{figure}

In this illustration, the output nodes of the inner code, denoted by \( i = 1 \ldots 7 \), is connected to the corresponding input nodes of the outer code, \( i' = 1 \ldots 7 \). By sequentially applying the GLC—which is a specific instance of the (\(\color{red}{pv}_1\)) rule—between pairs \( (i, i') \), we can derive the encoder diagram for the concatenated graph code. The detailed steps of this process are elaborated in Appendix \ref{ap:713}.

We will first employ the ZX encoder diagram to re-examine and affirm the proposition outlined in Theorem \ref{thm:glc}. Subsequently, we will showcase scenarios where the GLC might not suffice, yet the ZX encoder diagram proves to be an effective tool for analyzing concatenated graph codes.

\subsection{Re-proof of Theorem \ref{thm:glc} Utilizing the ZX Encoder Diagram}

To corroborate Theorem \ref{thm:glc}, we initiate with a rudimentary case involving a single qubit \(i\) from the outer code concatenated to qubit \(i'\) from the inner code. It is presupposed that initially, \(i\) and \(i'\) lack common neighbors, a premise based on the stacking of encoders. (It's noteworthy that the absence of edges between \(N(i) \setminus \{i'\}\) and \(N(i') \setminus \{i\}\) is presumed, although its inclusion doesn't impinge on the proof.) The pivot rule (\(\color{red}{pv}_1\)) is then applied between \(i\) and \(j\). This action induces a GLC between \(N(i) \setminus \{i'\}\) and \(N(i') \setminus \{i\}\).
%where \(N(i)\) and \(N(j)\) denote the neighbors of \(i\) and \(j\) within the outer and inner codes, respectively. 
This procedure is graphically represented in Fig. \ref{fig:glc_proof}.

\begin{figure}[H]
\centering
  \includegraphics[width=3.5in]{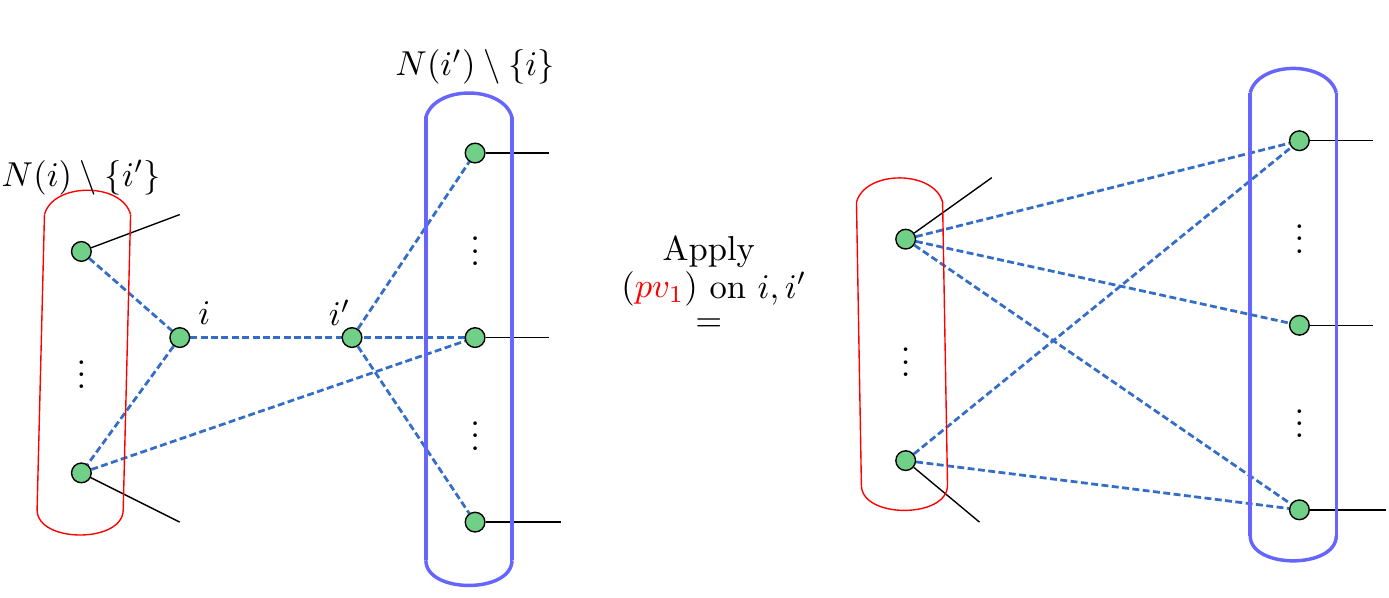}
  \caption{Visualization of GLC through the application of the pivot rule ($\color{red}pv_1$) between $i$ and $i'$.}
  \label{fig:glc_proof}
\end{figure}

This principle extends to scenarios with multiple qubits from the outer code concatenated to the inner code. Consider, for example, two qubits, \(i\) and \(j\) from the outer code, concatenated to qubits \(i'\) and \(j'\) in the inner code, respectively. We assume \(i\) and \(j\) are interconnected by a Hadamard edge, as the absence of such a connection simplifies the case to that of a single qubit discussed previously. Crucially, the neighbor sets within the inner code for these qubits, \(N(i')\) and \(N(j')\), must be disjoint, i.e., \(N(i') \cap N(j') = \emptyset\).

\begin{figure}[H]
  \centering
  \includegraphics[width=4.5in]{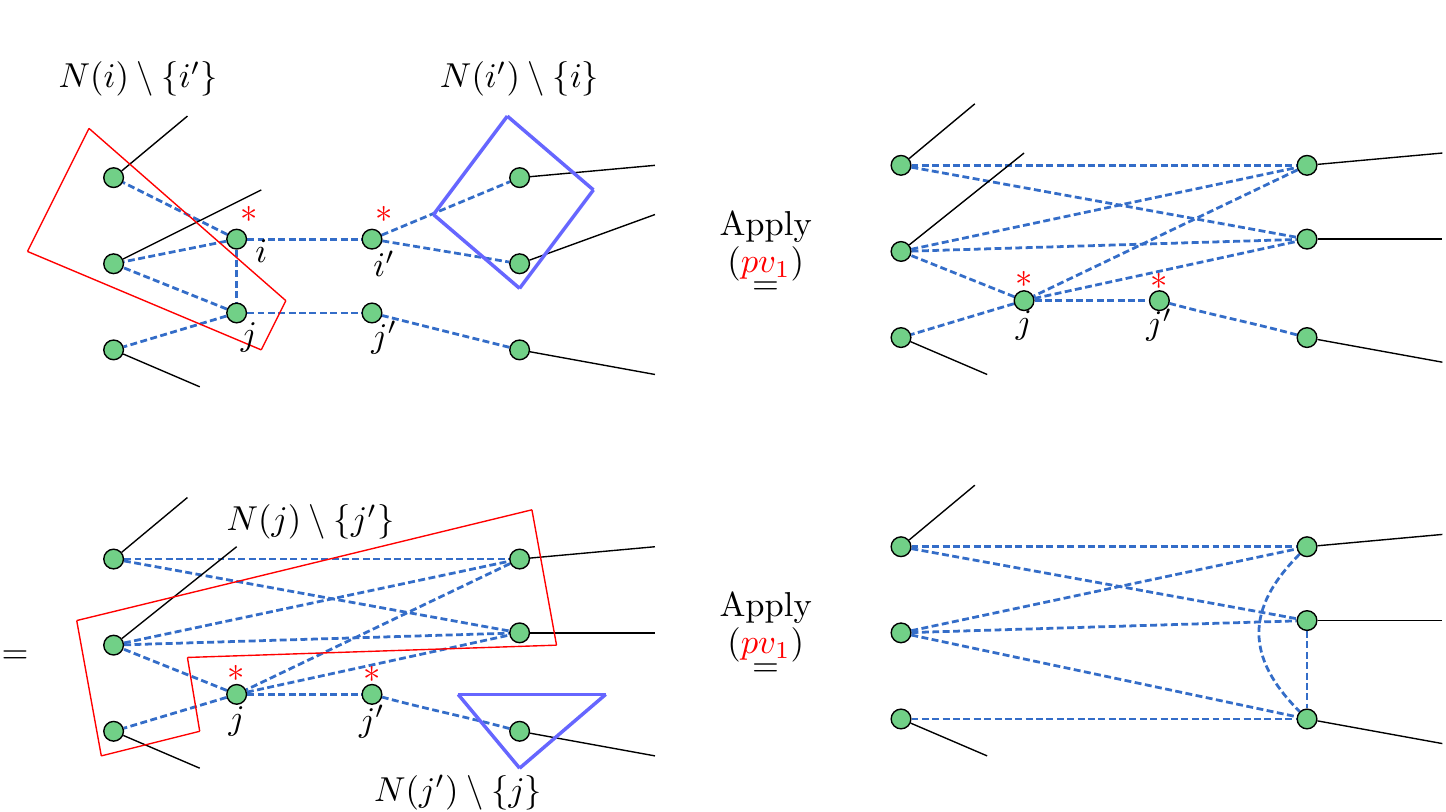}
  \caption{GLC analysis for concatenated graph codes with two qubits \(i\) and \(j\) from the outer code concatenated to \(i'\) and \(j'\) in the inner code, assuming disjoint neighbor sets \(N(i')\) and \(N(j')\). The figure demonstrates the sequential pivotting rule application and resultant GLC operations.}
  \label{fig:glc_analysis}
\end{figure}

Initially, the rule (\(\color{red}{pv}_1\)) is applied to the pair \((i, i')\), effectuating a GLC between \(N(i) \setminus \{i'\}\) and \(N(i') \setminus \{i\}\). Because at the beginning, \(N(i')\cap N(j') =\emptyset\), after pivoting on \((i, i')\), the sets \(N(j)\cap N(j')\) remain disjoint as well. Consequently, we can apply the pivot rule (\(\color{red}{pv}_1\)) to the pair \((j, j')\), resulting in a GLC between \(N(j)\setminus\{j'\}\) and \(N(j')\setminus\{j\}\). Figure \ref{fig:glc_analysis} illustrates this process.
The described methodology can be generalized to any number of qubits connected to the inner code, provided that the condition \( N(i') \cap N(j') = \emptyset \) holds, thus completing the proof of Theorem \ref{thm:glc} for concatenated graph codes.

%In a general setting, while application of ($\color{red}lc_1$),($\color{red}pv_1$) and ($\color{red}pv_2$)  confirms the feasibility of removing all interior nodes to streamline the encoder diagram, the practical implementation of simplification rules offers significant leeway, such as selecting node pairs and their processing order. Automating this simplification process rarely produces an "ideal" encoder diagram for concatenated codes. 
From the perspective of simplifying ZX-diagrams, the GLC provides guidance by identifying appropriate node pairs for applying the pivot rule ($\color{red}pv_1$). This helps to maintain the symmetry of the diagram, as demonstrated in Fig.\ref{fig:713_concate}. In comparison, non-targeted simplification can lead to an encoder that is complex in terms of its edges and connectivity. Such complexity may obscure the analysis of the encoder's symmetry and make the extraction of the encoding circuit more complicated.

The limitations of GLC are also apparent. It is applicable only in the case where the inner code is concatenated with a basis change $U'=H^{\otimes k'}$, where $i$ from the outer code and $i'$ from the inner code are connected by a Hadamard edge. In comparison, the ZX-calculus guarantees simplification of the diagram as long as $U'$ is a Clifford operation. Moreover, as we will see in the subsequent subsection, even when we choose $U'=H^{\otimes k'} $, there can be situations where \( N(i') \cap N(j') \neq \emptyset \), rendering Theorem \ref{thm:glc} invalid. Such situations can also be addressed by the ZX-calculus approach.

\subsection{Failure of GLC}
The ZX encoder diagram proves particularly useful in instances where the GLC rule fails. This failure occurs notably when \(N(i') \cap N(j') \neq \emptyset\). Let's consider the process step-by-step:

Initially, we apply the pivot rule (\(\color{red}pv_1\)) to qubits \(i\) and \(i'\), resulting in a GLC between \(N(i)\setminus\{i'\}\) and \(N(i')\setminus\{i\}\). However, since \(N(i') \cap N(j') \neq \emptyset\) and \(j \in N(i)\), the initial pivoting leads to the intersection \(N(i') \cap N(j')\) becoming the common neighbor set \(N(j) \cap N(j')\) for both \(j\) and \(j'\). This intersection disrupts the standard GLC process because it introduces an overlap between the two clusters that should remain separate to achieve a successful GLC.

\begin{figure}[H]
  \centering
  \includegraphics[width=4.5in]{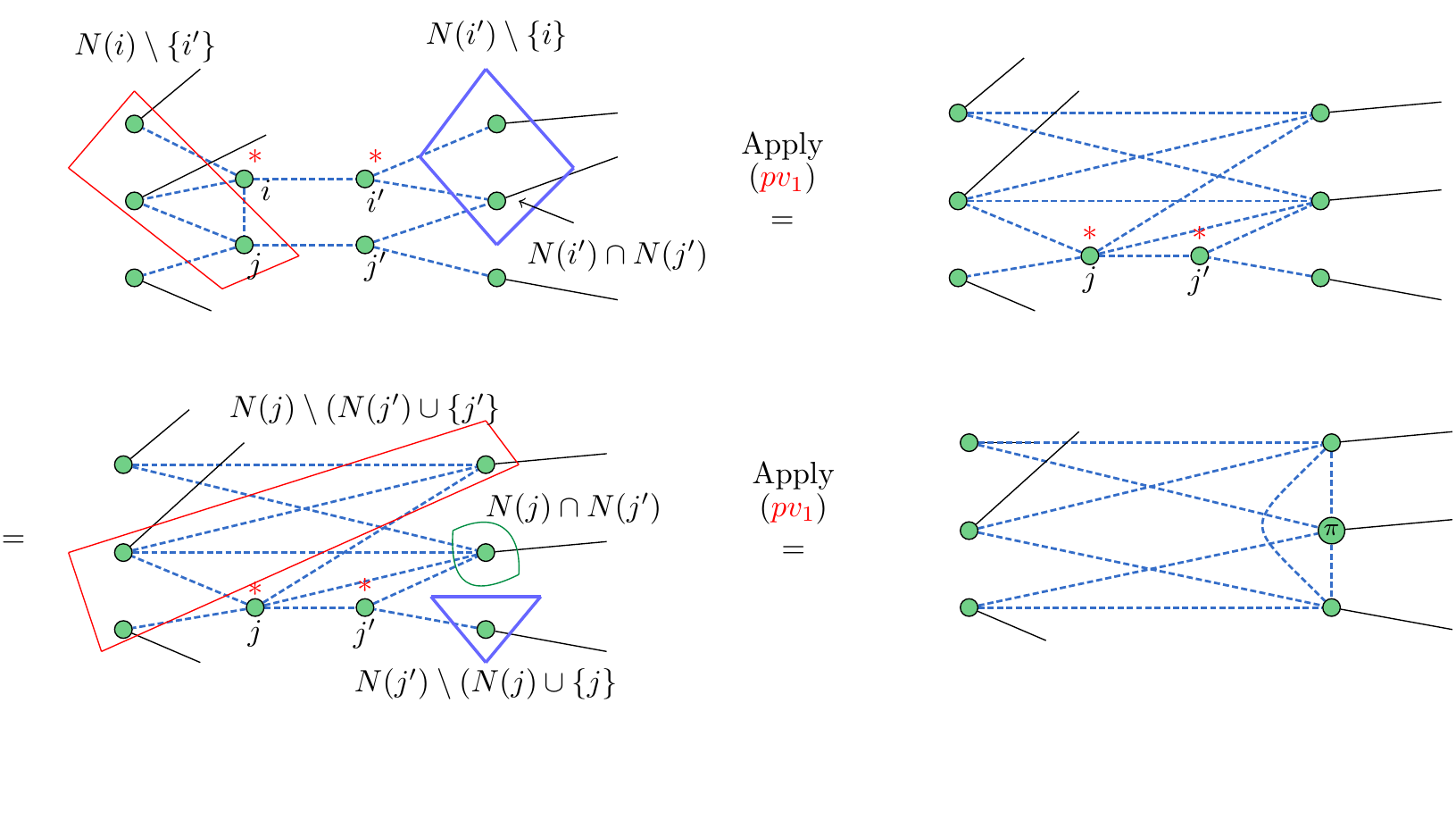}
  \caption{Depiction of GLC failure in concatenated graph codes with overlapping neighbor sets \(N(i') \cap N(j')\), necessitating additional GLC operations and \(\pi\) phase adjustments for the overlapping nodes.}
  \label{fig:glc_analysis_2}
\end{figure}

Addressing the overlap of neighbor sets, the rule (\(\color{red}pv_1\)) is applied between nodes \(j\) and \(j'\). %This intricate case is more effectively managed with the (\(\color{red}pv_1\)) rule rather than strict GLC, incorporating additional GLC operations and \(\pi\) phase adjustments as depicted in Fig.\ref{fig:glc_analysis_2}. 
In this context, for the application of the (\( \color{red}{pv_1} \)) rule as illustrated in Fig.~\ref{fig:pivot}, the sets \(N(j) \setminus (N(j') \cup \{j'\})  \) and \(N(j') \setminus (N(j) \cup \{j\})  \) correspond to the \( \alpha \) and \( \gamma \) clusters, respectively. Meanwhile, the intersection \( N(j) \cap N(j') \) is associated with the \( \beta \) cluster. The effects of applying the (\( \color{red}{pv_1} \)) rule are as follows:
\begin{itemize}
  \item GLC between \(N(j) \setminus (N(j') \cup \{j'\}) \) and \(N(j') \setminus (N(j) \cup \{j\}) )\)
  \item GLC between  \(N(j) \setminus (N(j') \cup \{j'\}) \) and \(N(j) \cap N(j')\)
  \item GLC between \(N(j') \setminus (N(j) \cup \{j\}) \) and \(N(j) \cap N(j')\)
  \item Addition of \(\pi\) phase to $Z$ spiders in \(N(j) \cap N(j')\)
\end{itemize}
The flexibility of the ZX encoder diagram is pivotal, particularly when traditional GLC is not suitable. This versatility accentuates the diagram's capacity to dissect and streamline the structure of concatenated graph codes. In instances where GLC is constrained, the ZX-calculus proves to be a formidable and perceptive substitute for examining the intricacies of graph code structures.

In summary, our study has highlighted the effectiveness of the ZX encoder diagram as an indispensable tool for analyzing concatenated graph codes, offering a more intuitive and visually guided approach compared to traditional algebraic methods. By applying the principles of Theorem \ref{th:simplification}, we are able to refine encoder diagrams, eliminating interior nodes to yield a concise representation that directly informs the structure of the resulting stabilizer code and the associated logical operations. This not only facilitates the construction of graph codes but also enhances our understanding of their inherent complexities, effectively marrying the abstract algebraic concepts with their graphical counterparts.

\section{Constructing Stabilizer Codes via Encoder Contraction}
\label{sec:fusing}

This section elucidates the utility of the ZX encoder diagram in the domain of QECCs construction. It highlights the innovative application of the ZX encoder diagram for constructing larger stabilizer codes from smaller foundational units, offering a novel approach beyond traditional concatenation methods.

At the core of this methodology is the conceptualization of the encoder diagram as an essential building block. By systematically layering these smaller units, the construction of more complex codes becomes achievable. This approach shares conceptual similarities with the contraction processes observed in quantum lego codes~\cite{cao2022quantum} and tensor network codes~\cite{farrelly2021tensor}. These methodologies provide a graphically intuitive and flexible framework for assembling sophisticated quantum error correction codes by combining simpler codes or states in a modular fashion.

% Uncomment the following paragraph to include it in the document.
% The ZX encoder diagram facilitates the representation of complex coding architectures as tensor networks. These networks are composed of tensors, which correspond to simple codes or states, arranged in a systematic and modular manner. This technique extends the traditional concept of code concatenation, envisioning the assembly of intricate codes from modular components.

Despite the foundational similarities with other graph-based codes, the employment of ZX encoder diagrams is distinguished by its computational efficiency and the clarity it brings to the construction process. This distinction underscores the ZX encoder diagrams' potential to significantly advance the development and understanding of quantum error correction codes.

As delineated previously, the encoder \(\mathcal{E}\) of a stabilizer code can be conceptualized as a stabilizer state \(|\mathcal{E}\rangle\), which emerges from applying the map-state duality transformation. To initiate with a generic representation, consider an \(m\)-qubit stabilizer state \(|\psi\rangle\), expressible as:
\begin{equation}
    |\psi\rangle = \sum_{i_1\ldots i_m} \psi_{i_1\ldots i_m}|i_1\ldots i_m\rangle,
\end{equation}
where \(\psi_{i_1\ldots i_m}\) denotes a rank-\(m\) tensor. The contraction of two tensor indices (illustratively, the first two indices) equates to projecting the corresponding qubits onto the unnormalized Bell state \(|\Phi^+\rangle = |00\rangle + |11\rangle\). The associated projection operator is defined as \(P^+ = |\Phi^+\rangle \langle \Phi^+| \otimes I^{\otimes (m-2)}\). Post-projection, the state transforms as follows:
\begin{eqnarray}
P^+|\psi\rangle &=& \sum_{j, k=0}^1 \sum_{i_1, \ldots, i_m=0}^1 |jj\rangle \otimes \langle kk|i_1 i_2\rangle \psi_{i_1 \ldots i_m}|i_3 \ldots i_m\rangle \\ 
&=& |\Phi^+\rangle \otimes \sum_{i_3, \ldots, i_m, k=0}^1 \psi_{kk i_3 \ldots i_m}|i_3 \ldots i_m\rangle \\ 
&=& |\Phi^+\rangle \otimes |\psi^*\rangle,
\end{eqnarray}
where \(|\psi^*\rangle\) denotes the resultant state following the contraction of indices. Notably, the physical realization of projecting two qubits into the Bell state involves post-selecting the measurement outcomes of $XX$ and $ZZ$ operators to be $+1$ for the corresponding qubits in \(|\psi\rangle\). This ensures \(|\psi^*\rangle\) remains a stabilizer state, albeit with updated stabilizers, adhering to the modifications delineated in Chapter 10.5 of \cite{nielsen2002quantum}.

The conventional representation of a stabilizer state using a tensor \(\psi_{i_1\ldots i_m}\), which requires \(2^m\) complex numbers, is not the most efficient method.A more effective approach involves representing the stabilizer state (encoder) through its stabilizer group, as referred to in the methodology employed in tensor network codes \cite{farrelly2021tensor}. For a \([[n,k,d]]\) stabilizer code, the representation is refined by defining a rank-\(n\) tensor to "one-hot" encode the elements of the stabilizer group \(\mathcal{S}\):
\begin{equation}
    T_{g_1\ldots g_n}(\overline{P}) =  \begin{cases}1 & \text{if } \sigma^{g_1} \otimes \ldots \otimes \sigma^{g_n} \in \mathcal{S} \overline{P} \\ 0 & \text{otherwise}\end{cases},
\end{equation}
for each logical Pauli \(\overline{P}\), where \(g_i \in \{0,1,2,3\}\), and \(\sigma^0= I, \sigma^1= X, \sigma^2= Y, \sigma^3= Z\). The contraction of \(T_{g_1\ldots g_n}\) also corresponds to projecting the associated qubit into the unnormalized Bell state.

The discussion extends to the ZX encoder diagram, which offers a valid tensor network representation of the stabilizer state. Crucially, specifying an \(m\)-qubit stabilizer state necessitates only \(m\) stabilizer generators due to the Abelian nature of the stabilizer group, suggesting that the ZX encoder diagram is an "efficient" tensor network for representing these generators.
\begin{figure}[H]
\centering
\includegraphics[width=2.4in]{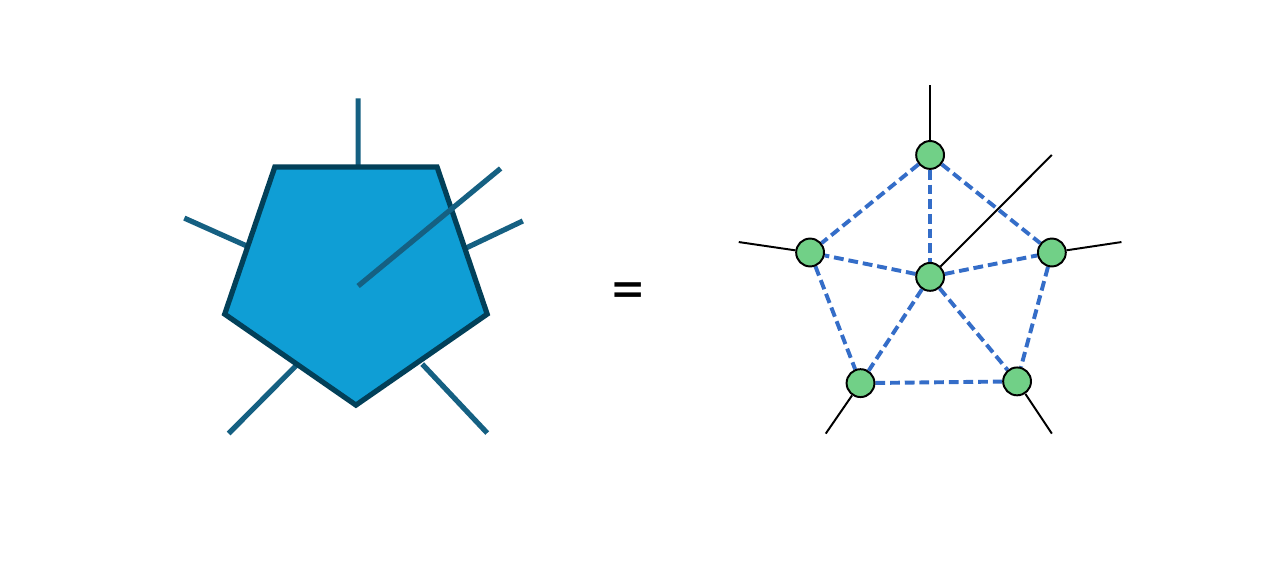}
 \caption{Illustration of the equivalence between tensor and ZX-diagrams, representing the 6-qubit Absolutely Maximally Entangled (AME) state\cite{helwig2013absolutely} or the encoder of the \([[5,1,3]]\) graph code via map-state duality.}
\label{fig:tensor_equivalent}
\end{figure}
Considering the encoder \(\mathcal{E}_G\) for a standard form \([[n,k,d]]\) graph code, the stabilizer generators \(S_{\mathcal{E}_G}\) for the corresponding graph state \(|\mathcal{E}_G\rangle\) are concisely represented as:
\begin{equation}
  S_{\mathcal{E}_G} = \left[\begin{array}{cc|cc}
    I_k & 0 & 0 & \Gamma  \\
    0 & I_n & \Gamma^T & G  \\
    \end{array}\right],
\end{equation}
where the qubits are indexed such that the first \(k\) function as input nodes, and the remaining \(n\) as output nodes. %The right half of \(S_{\mathcal{E}_G}\), shaped \((n+k)\times (n+k)\), aligns with the adjacency matrix of the associated graph of the graph state \(|\mathcal{E}_G\rangle\).
Thus, the encoder diagram employs the graph \(G\) and classical codeword \(\Gamma\) to "store" the stabilizer generators of the stabilizer state.

This framework can be readily generalized to stabilizer codes with minor adjustments. The stabilizer generators for the stabilizer code \(|\mathcal{E}\rangle\) must be conjugated by the Clifford \(\mathcal{C}l\) and local Clifford operations \(L_c\), as outlined in Proposition \ref{pro:encoder_stb}. Consequently, the ZX encoder diagram can be viewed as utilizing minimal information to characterize the stabilizer state while maintaining the graph's structure in a tensor-network representation.

We demonstrate that the contraction of the ZX encoder diagram is equivalent to the contraction of a stabilizer state. This equivalence is straightforward, as both contractions project corresponding qubits onto the Bell state, symbolized by a cup in the ZX-calculus.

\begin{figure}[H]
\centering
\includegraphics[height=2.4in]{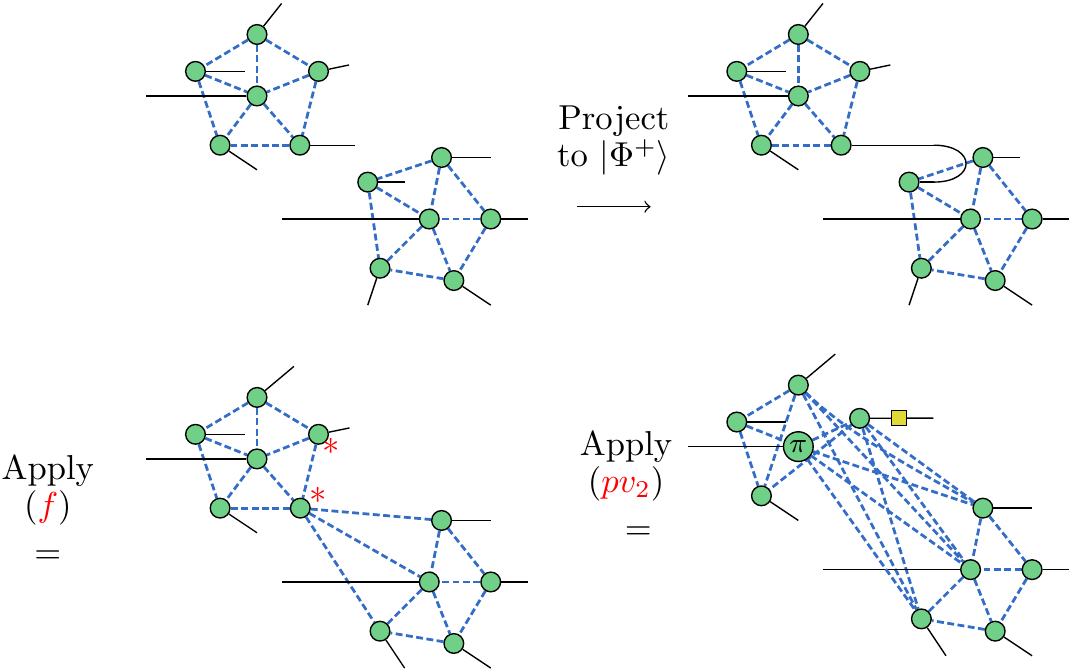}
\caption{Equivalence of tensor network and ZX encoder diagram contraction, illustrating the merging of two stabilizer encoders into a larger encoder using the [[5,1,3]] code as a concrete example. After simplification, the stabilizer and logical Pauli operators of the contracted code are obtained.}
\label{fig:fusing}
\end{figure}
After projection, to extract the stabilizer and logical Pauli operators of the contracted code from the ZX encoder diagram, we apply the simplification process outlined in Theorem \ref{th:simplification}, which involves the removal of all interior $Z$ spiders. The resultant diagram is locally Clifford equivalent to a graph code, from which stabilizers can be readily obtained. The validity of the newly constructed encoder as a stabilizer code encoder can be verified using Proposition \ref{pro:enc_valid}. 
The operational parallels of the ZX encoder diagram to tensor network contractions reduce computational complexity and enhance the graphical intuitiveness in the construction of stabilizer codes. The diagram encapsulates all relevant data for stabilizer codes, allowing for direct inference of logical states and operations, thereby streamlining code analysis and synthesis.

The contraction process within the ZX-calculus can be efficiently performed. When two legs in the diagram are connected by a solid line, it signifies either the application of a single qubit Clifford operation (potentially a \( Z \) spider with a \( \pi/2 \) phase, accompanied by an optional Hadamard gate) or a straightforward connection. Each scenario requires a constant number \( O(1) \) of pivoting or local complementation operations, enabling the contraction of the diagram in polynomial time.

It is noteworthy that the representation of an encoder in the ZX-calculus is not unique, owing to the non-uniqueness of the stabilizer group's generator set. The choice of pivot order and pivot pairs significantly influences the final encoder diagram, although all representations are mathematically equivalent. A "nicer" encoder diagram can substantially reduce the complexity of the resulting quantum circuit.

An intriguing application of tensor networks in quantum error correction is the construction of holographic codes. We elucidate how holographic codes can be constructed through ZX encoder contraction and interpreted within the graph code framework. This approach highlights the synergistic relationship between holographic codes, graph codes, and tensor network codes, enriching our understanding of their interconnectedness.

\subsection{Construction of HaPPY Codes}
Holographic codes, inspired by the AdS/CFT holographic duality in theoretical physics, represent a unique class of quantum error-correcting codes that emulate radial time evolution in Anti-de Sitter space~\cite{pastawski2015holographic,almheiri2015bulk}. These codes are crafted from a lattice of perfect tensors that tessellate hyperbolic space, establishing a correspondence between logical qudits within the AdS bulk and physical qudits at the boundary of the associated Conformal Field Theory (CFT). In this framework, physical qubits at the boundary are denoted by uncontracted tensor legs, while logical qubits within the bulk are signified by uncontracted legs in the bulk, offering a novel perspective on the intricate structure of holographic codes.

A conventional approach to constructing holographic codes involves utilizing tensor networks that incorporate absolutely maximally entangled (AME) states as fundamental building blocks. A notable example is the hyperbolic pentagon (HaPPY) code, assembled through the contraction of 6-qubit AME states. Specifically, the encoding mechanism of HaPPY codes employs a tensor network characterized by five-qubit encoding isometries. These isometries manifest as six-legged perfect tensors, with five legs corresponding to physical qubits and one leg dedicated to the encoded logical qubit. This configuration serves as a foundational model for the AdS/CFT holographic duality, illustrating the intricate relationship between quantum error correction and theoretical physics.

Recent work by Munne et al.~\cite{munne2022engineering} has illuminated the representation of the HaPPY code as a graph code, offering a new perspective on the structure of holographic codes. This study begins with an algebraic approach to derive the stabilizers of the HaPPY code and then performs local Clifford operations to construct the corresponding graph code. It leverages the algorithm proposed in Ref.~\cite{Adcock_2020} to enumerate locally Clifford-equivalent graph codes on a manageable scale of qubits. By exploring the local Clifford orbit, they identified an optimized graph characterized by a minimal number of edges and predominantly short-range interactions, which enhances the practicality of the code's implementation. However, the scalability of this method is limited by the exponential growth of the local Clifford orbit as the number of qubits increases, making it less feasible for larger code instances.
\begin{figure}[H]
  \centering
  \includegraphics[height=4.2in]{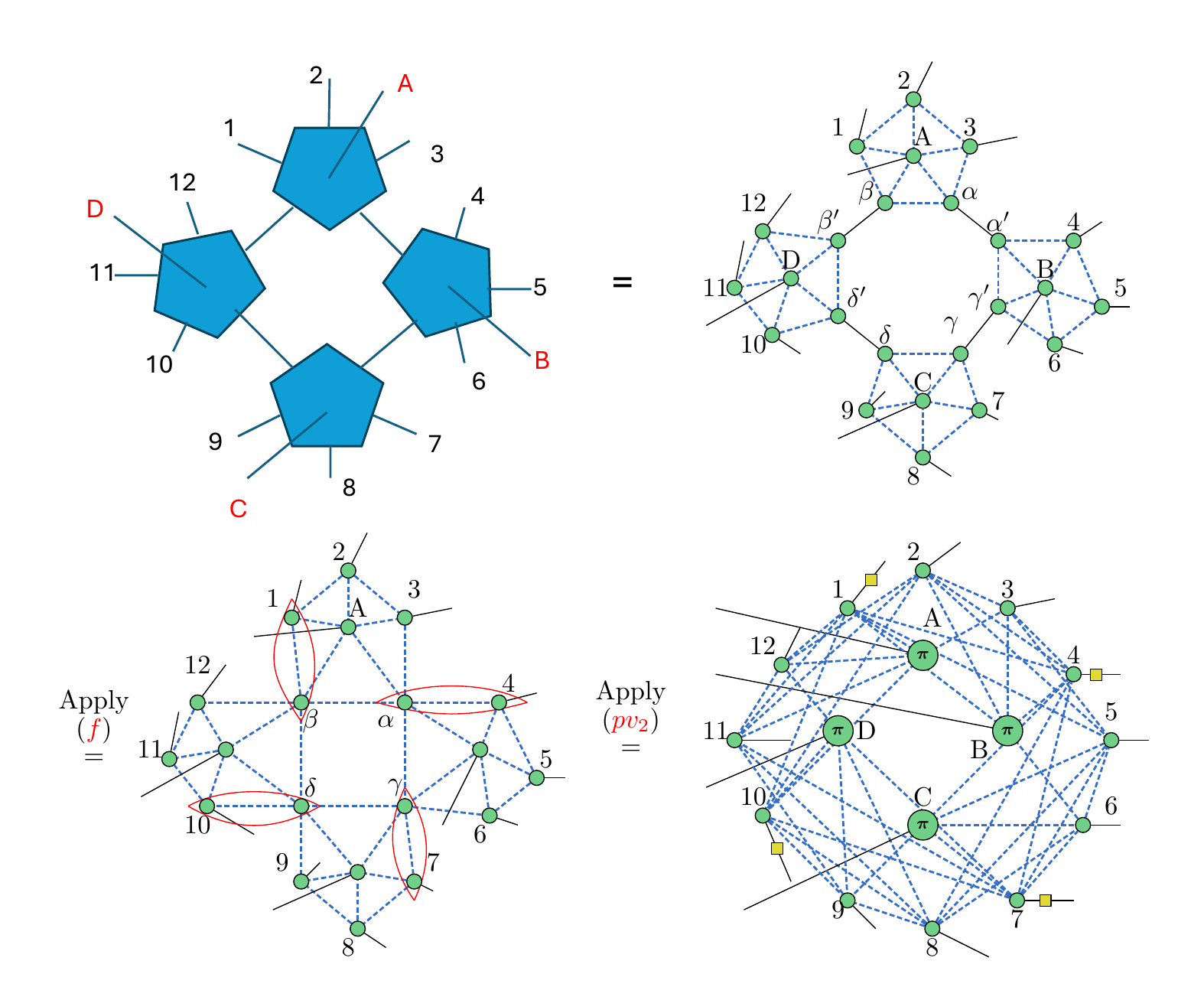}
  \caption{Construction of the HaPPY code by contracting four AME tensors. In the language of the ZX-calculus encoder diagram, it encodes 4 qubits (\(ABCD\)) into 12 qubits (\(1-12\)). Red loops denote pairs of $Z$ spiders where the ($\color{red}pv_2$) rule is applied. The step-by-step derivation is given in Appendix \ref{ap:Happy}.}
  \label{fig:holo_step}
\end{figure}
To circumvent this limitation, we explore the application of the ZX encoder diagram, employing the ZX-calculus as a powerful tool for this context. Specifically, we examine an encoding scenario of 4 bulk qubits into 12 boundary qubits, the same as in Ref.~\cite{Adcock_2020}. By adopting a symmetric pivot sequence, with the ($\color{red}pv_2$) rule, we successfully derive an encoder diagram that exhibits optimization comparable to that achieved by \cite{munne2022engineering}.
The process of contracting six-legged perfect tensors unfolds as follows:
\begin{itemize}
    \item Encode qubit $A$ into a five-qubit graph code.
    \item Encode qubit $\alpha$ alongside qubit $B$ into qubits $4, 5, 6$, and $\gamma'$.
    \item Encode qubit $\gamma'$ together with qubit $C$ into qubits $7, 8, 9$, and $\delta$.
    \item Encode qubits $\delta$ and $\beta$, along with qubit $D$, into qubits $10, 11, 12$, culminating in a fully encoded holographic code.
\end{itemize}
The ZX encoder simplification process, depicted in FIG.~\ref{fig:holo_step}, introduces only four interior nodes during the contraction. These nodes can be effectively eliminated by employing ($\color{red}pv_2$).

With the ZX encoder diagram established, the encoding circuit, stabilizers, and logical Pauli operators can be easily deduced. This process involves a methodical approach where we initially focus on the encoder diagram, which visually encapsulates the functionalities of the encoding circuit. 

To derive a local Clifford equivalent graph code, we remove the \(Z\) gates on input nodes \(A\), \(B\), \(C\), and \(D\), and the Hadamard gates on the output nodes \(1\), \(4\), \(7\), and \(10\).
The graph state $|\overline{0\ldots0}\rangle_G$ of the graph code, as shown in FIG.~\ref{fig:holo2}, is locally Clifford equivalent to the results reported in Ref.~\cite{munne2022engineering}. 
\begin{figure}[H]
    \centering
      \includegraphics[width=4.3in]{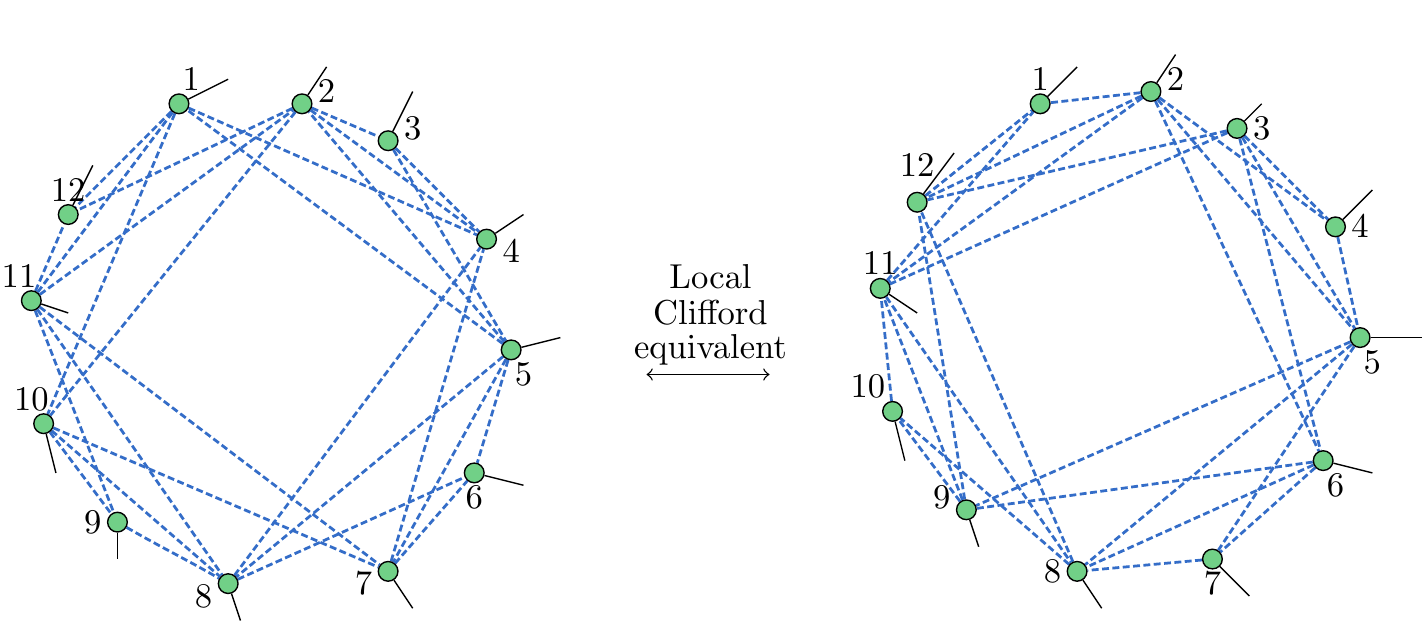}
      \caption{The graph state $|\overline{0\ldots0}\rangle_G$ of the holographic graph code is depicted. The left image presents the graph state derived from the ZX encoder diagram after removing $\mathcal{C}l$ and $L_c$, showcasing the graph code. The right image shows the optimized graph state as reported in Ref.~\cite{munne2022engineering}. These two graph states are locally Clifford equivalent.}
      \label{fig:holo2}
\end{figure}

The ZX-calculus provides a suitable language for this study, particularly when analyzing the role of symmetry in graph codes. As we previously discussed, the encoder has exponentially many equivalent representations, most of which do not exhibit symmetry. The symmetric encoder diagram, obtained through a sequence of symmetric pivoting as shown in Figure~\ref{fig:holo_step}, is likely to correspond to "nice" encoder that require short-range interactions and have few connections. However, the impact of this symmetry on the efficiency of the encoding requires further investigation.

Transitioning from the specifics of the encoder diagram to the broader context of circuit simplification, it is important to recognize the underlying efficiencies brought forth by the ZX-calculus in this domain. Moreover, the simplification process, as defined by Theorem~\ref{th:simplification}, is more efficient than tensor network contraction\cite{farrelly2021tensor} and as efficient as the operator pushing in\cite{cao2022quantum} since it only deals with the stabilizer generator. The advantages of ZX-calculus include the ability to represent any quantum circuit, with the extraction of diagrams to circuits being well-developed. Circuit identities are already covered by the rules of ZX-calculus, showcasing its power in simplifying Clifford\cite{Duncan_2020} and universal quantum circuits (reducing T gates)\cite{Kissinger_2020}. The transition from the specific case of the encoder diagram to these broader applications demonstrates the versatility and comprehensive nature of the ZX-calculus as a tool for optimizing quantum computations across various scenarios.
 Simplification of encoding and decoding circuits can be readily studied using this powerful diagram representation to satisfy desired metrics. The future work in this domain will likely further harness the potential of ZX-calculus, especially in understanding the relationship between symmetry and circuit efficiency.
%The versatility of the ZX-calculus as a tool for optimizing quantum computations is evident. It allows for the simplification of encoding and decoding circuits to meet desired performance metrics. The future work in this domain will likely further harness the potential of ZX-calculus, especially in understanding the relationship between symmetry and circuit efficiency.

 %
 
\section{Conclusion and discussion}

In conclusion, our work effectively leverages the ZX-calculus to advance our understanding and construction of quantum error-correcting codes. Through the development of the ZX encoder diagram, as discussed in Section~\ref{sec:zxencoder}, we've provided a unified representation of both graph and stabilizer codes. This has enabled us to simplify their complexity and extend their applicability in quantum computing. In Section~\ref{sec:Concatenation}, we revisited and expanded the principles of concatenated graph codes, and in Section~\ref{sec:fusing}, we demonstrated the practicality of this approach in constructing complex codes, exemplified by the HaPPY code.

The insights gained from this research not only enhance our understanding of the complex combinatorial structures of graph codes and their concatenations but also suggest potential areas for future exploration. The investigation of symmetries within constructed codes, as observed in our studies, hints at more efficient error detection and correction techniques. Additionally, understanding these symmetries could unveil connections between seemingly distinct codes, enriching our comprehension of the quantum error-correcting code space.

Our research thus underscores the significant role of ZX-calculus as both a theoretical and practical tool in quantum error correction. By continuing to leverage this graphical language, we aim to contribute further to the development of robust and efficient quantum computing systems.
\section*{Acknowledgement} We thank Markus Grassl ,Chao Zhang and Sarah Meng Li for helpful discussions. This work is support by GRF (grant no. 16305121) and the National Science Foundation of China (Grants No. 12004205).

\appendix

\section{Circuit Extraction from Encoder Diagram}
\label{ap:encoder_circuit}

The process of extracting a quantum circuit from an encoder diagram involves the application of specific ZX-calculus rules pertinent to encoder transformations. One such rule is the ($\color{red}cx$) rule, which facilitates the manipulation of the bi-adjacency matrix \(\Gamma\) between input and output nodes. This rule, as delineated in Proposition 7.1 by \cite{Duncan_2020}, is instrumental in describing the effect of applying a CX gate within the graph structure.
\begin{figure}[H]
    \centering
    \includegraphics[width=2.6in]{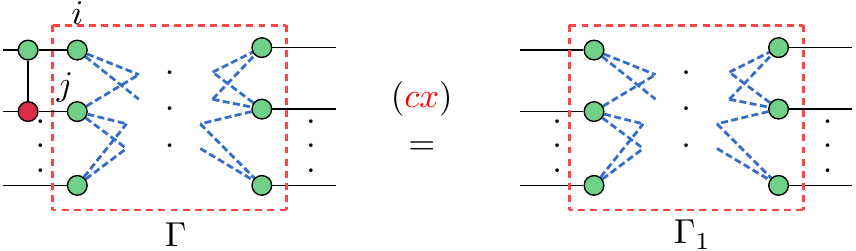}
    %\caption{Illustration of the ($\color{red}cx$) rule applied to the encoder diagram.}
    \label{fig:cx_rule}
\end{figure}
When a CX gate is applied with qubit \(i\) as the control and qubit \(j\) as the target, the resultant bi-adjacency matrix, denoted \(\Gamma_1\), evolves from \(\Gamma\) by adding row \(i\) to row \(j\). This operation reflects the modifications induced in the bipartite graph's structure. Conversely, applying a CX gate on the right side of the encoder affects the columns, resulting in a modified bi-adjacency matrix \(\Gamma_2\), obtained by adding column \(i\) to column \(j\).

\begin{proof}
\begin{figure}[H]
    \centering
    \includegraphics[width=4.4in]{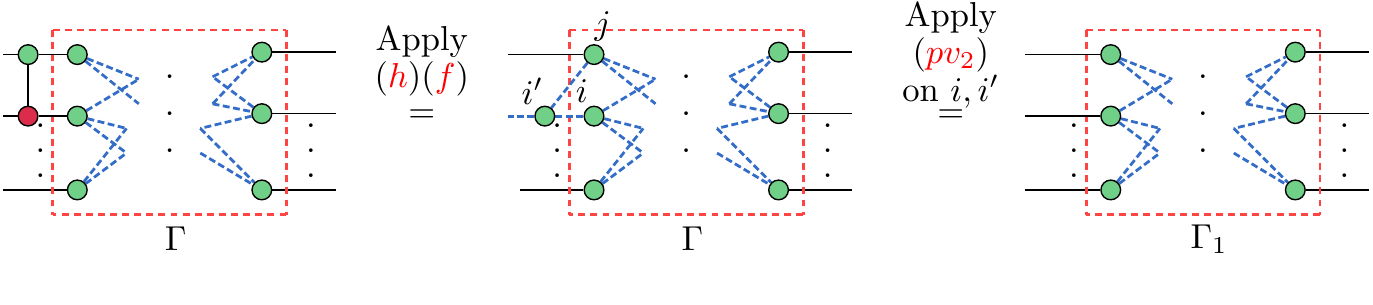}
    \caption{Proof of the ($\color{red}cx$) rule's effect on the encoder diagram.}
    \label{fig:cx_rule_proof}
\end{figure}
The application of the (\( \color{red}{pv_2} \)) rule to nodes \( i \) and \( i' \) introduces a Hadamard edge between every node in \( N(i) \setminus{i'}\) and node \( j \). This operation, in the context of the bi-adjacency matrix, results in \(\Gamma_1\), which is derived from \(\Gamma\) by adding row \(i\) to row \(j\). The application of the ($\color{red}cx$) rule on the encoder's right side analogously modifies the columns.
\end{proof}

By executing column operations (applying ($\color{red}cx$) on the right hand side),  the matrix $\Gamma$ can be transformed into the form $\left[I_k \ \ 0\right]$ (up to row swapping). This transformation systematically facilitates the extraction of the corresponding quantum circuit from the encoder diagram, as illustrated below:
\begin{figure}[H]
    \centering
    \includegraphics[height=2.3in]{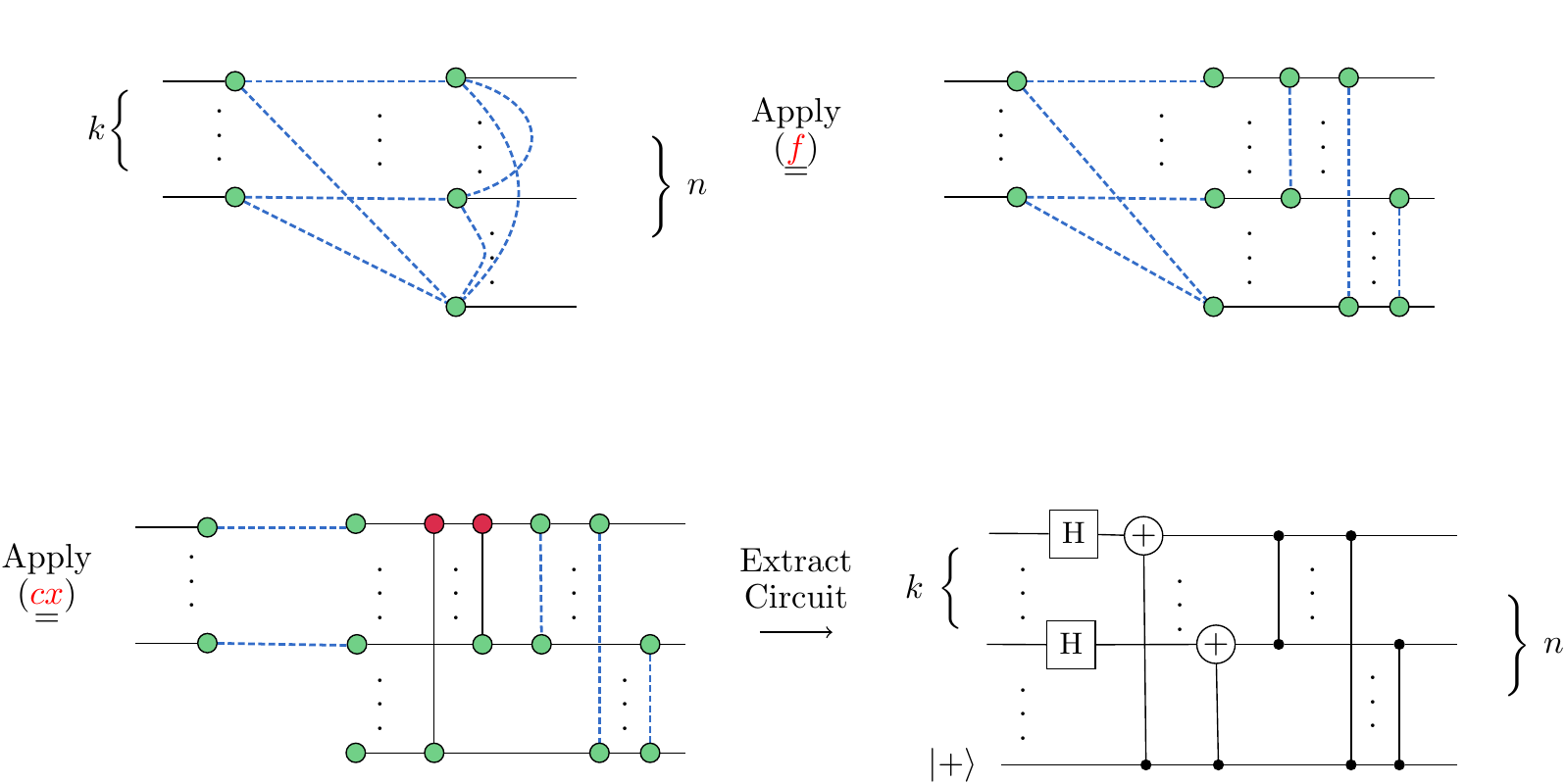}
    \caption{Systematic extraction of the quantum circuit from the encoder diagram.}
    \label{fig:zx_to_circuit}
\end{figure}

\section{Extraction of  Stabilizers and Logical Pauli from ZX Encoder Diagrams}
\label{ap:stab}

In this appendix, we delineate the methodology for deriving the logical Pauli operators and stabilizers from a given ZX encoder diagram. For a standard form graph code encoder $\mathcal{E}_G$, characterized by \(\Gamma\) in RREF and graph \(G\), the stabilizers \(\mathcal{S}_G\) and logical Pauli operators \(\bar{X}_G\) and \(\bar{Z}_G\) are determined as outlined in Eq.(\ref{eqn:check_matrix}).

To ascertain the stabilizer $\mathcal{S}$ of a stabilizer code correspond to encoder $\mathcal{E}$, the encoder can initially be transformed into the form \(\mathcal{E} = L_c \mathcal{E}_G \mathcal{C}l_*\) by applying CX to the left side of the encoder. This row operation modifies \(\Gamma\) into its RREF (assuming the qubits are properly indexed). Diagrammatically, this transformation is depicted as follows:

\begin{figure}[H]
    \centering
    \includegraphics[width=5in]{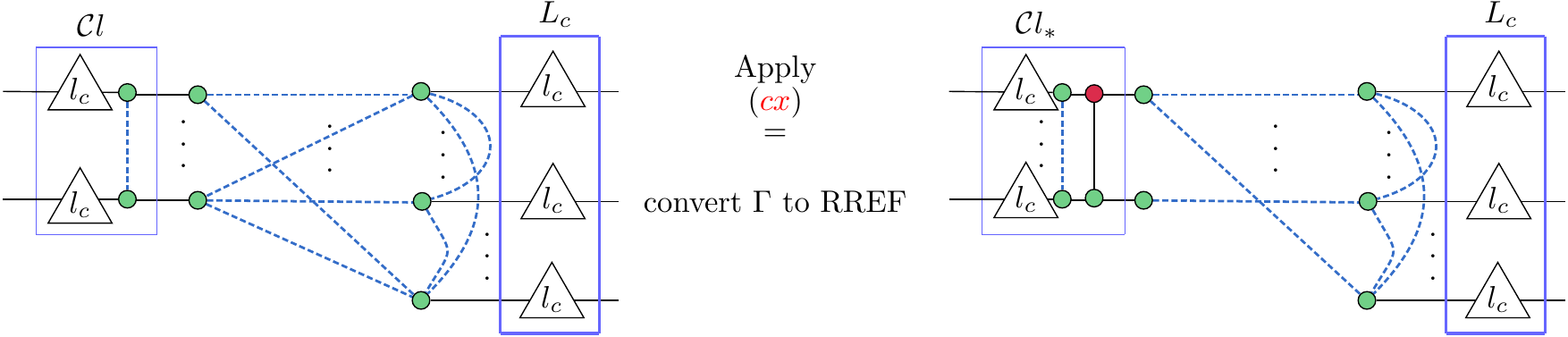}
    \caption{Transformation of the encoder into standard form.}
    \label{fig:encoder_transform}
\end{figure}

Given that the Clifford operation \(\mathcal{C}l_*\) does not alter the codespace, the stabilizer  \(\mathcal{S}\) can be expressed as:

\begin{equation}
    \mathcal{S} = L_c \mathcal{S}_G L_c^{\dagger}.
    \label{eqn:stab_transform}
\end{equation}
To identify the logical Pauli operators, it is essential to understand the mechanism of operator pushing through the encoder. If a logical operator \( L \) on the left side of the encoder is equivalent to an operator \( \overline{L} \) on the right side (i.e., \( \mathcal{E}L = \overline{L} \mathcal{E} \)), then \( \overline{L} \) is a valid physical realization of the logical operator \( L \). This equivalence is diagrammatically represented as:

\begin{figure}[H]
    \centering
    \includegraphics[width=3.8in]{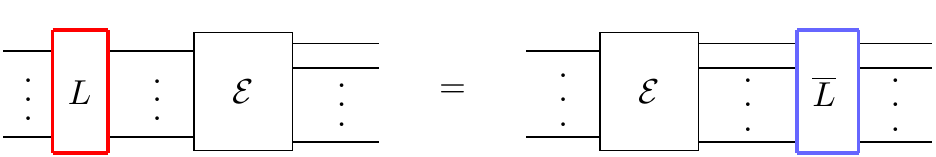}
    \caption{Pushing of logical operators through the encoder. Ref.~\cite{Garvie_2018} provides the proof.}
    \label{fig:logical_propagation}
\end{figure}
For \(\mathcal{E}_G\), the logical Pauli operators \(\overline{P_G}\in \mathcal{P}_n\) are outlined in Eq.(\ref{eqn:check_matrix}). Assuming the logical Pauli operator for \(\mathcal{E}\) is denoted by \(\overline{P}\in \mathcal{P}_n\), then:

\begin{equation}
    \overline{P} = L_c^{\dagger} \overline{P_G'} L_c,
    \label{eqn:logical_pauli_transform}
\end{equation}

where \(P' = \mathcal{C}l_*^{\dagger} P \mathcal{C}l_* \in \mathcal{P}_k\). This relationship can be visualized as follows:

\begin{figure}[H]
    \centering
    \includegraphics[width=5in]{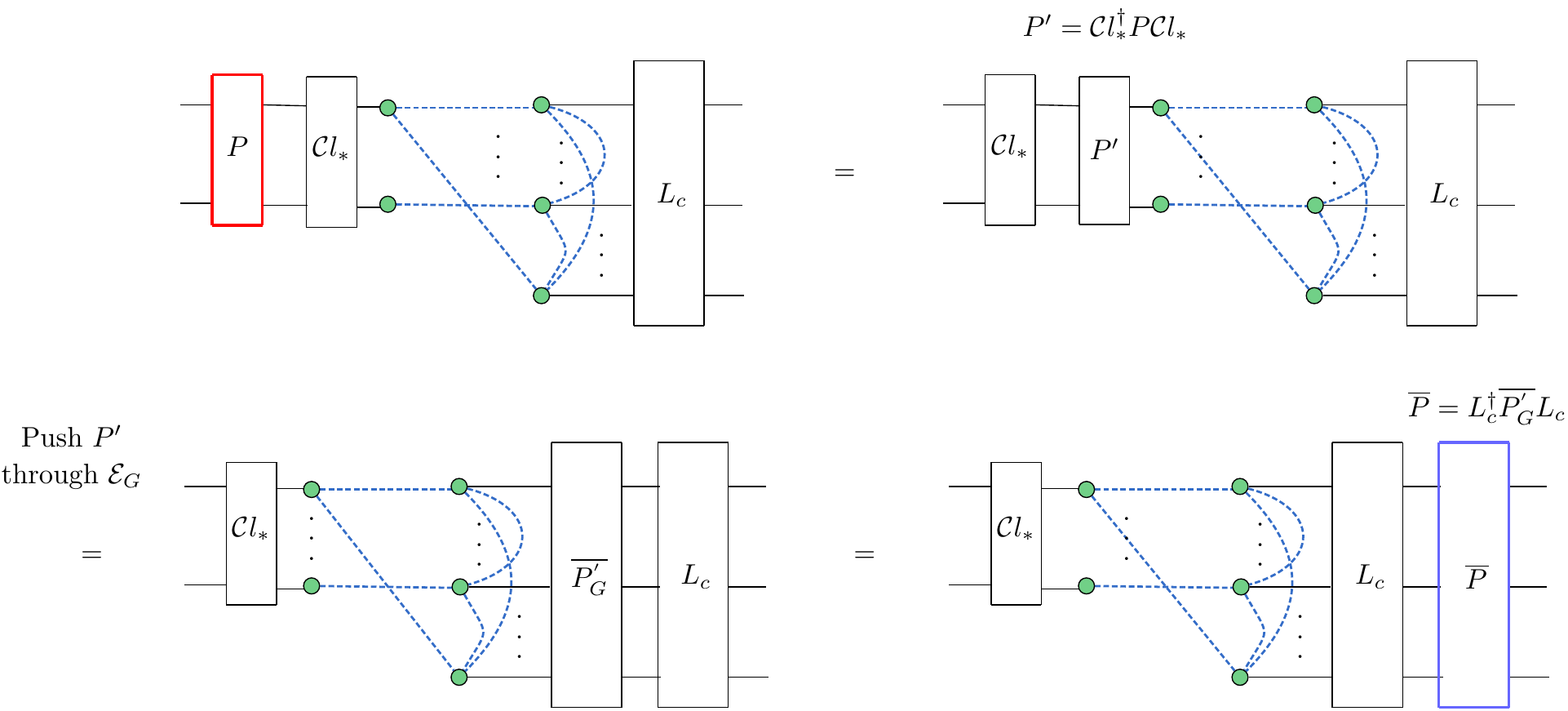}
    \label{fig:enter-label}
\end{figure}

\section{Simplification of self-concatenated [[7,1,3]] code}
\label{ap:713}
\begin{figure}[H]
    \centering
    \includegraphics[width=5in]{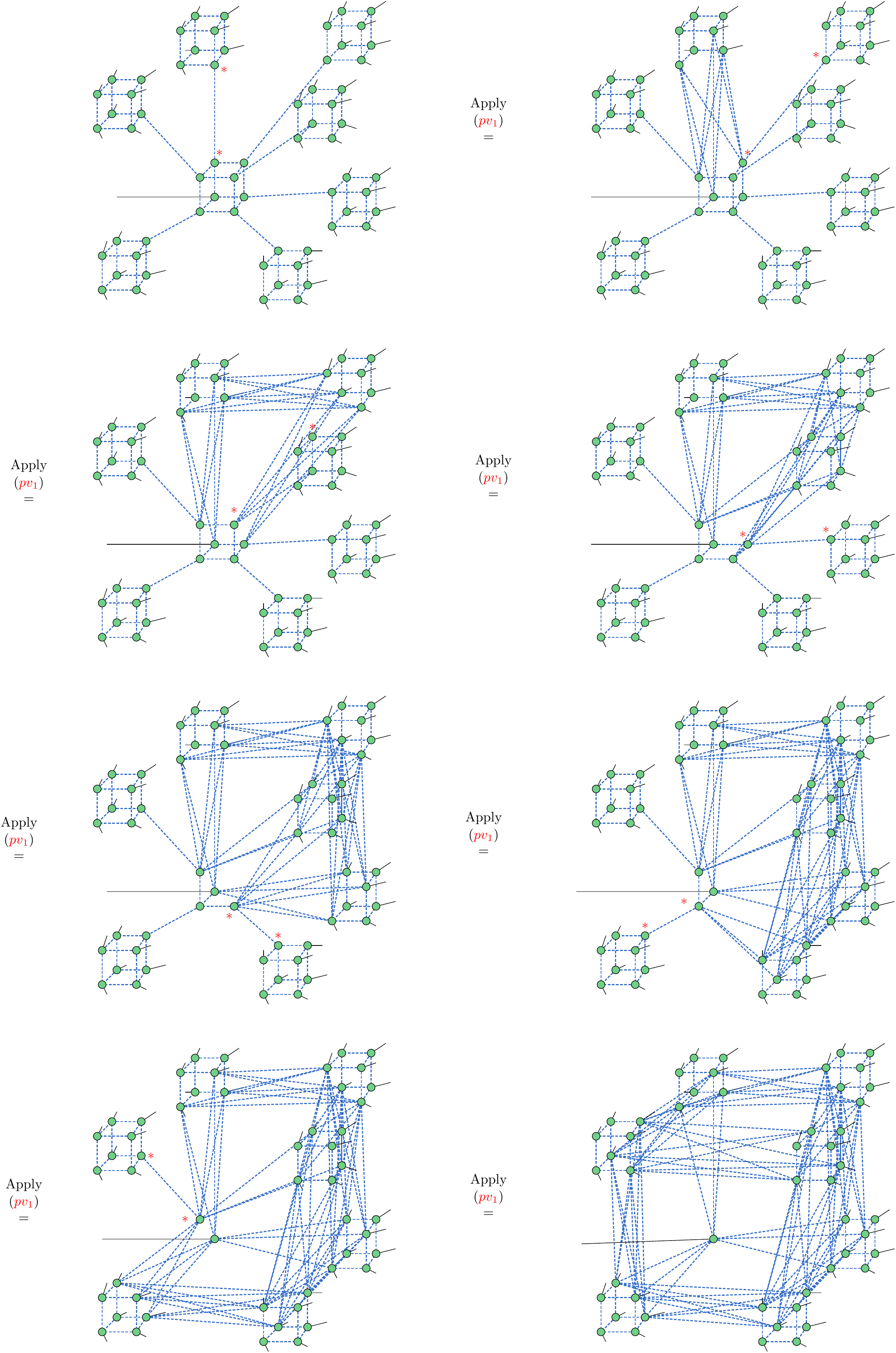}
    \caption{Self-concatenation of the \([[7, 1, 3]]\) graph code, resulting in a \([[49, 1, 9]]\) code.  The operation ($\color{red}pv_1$) is applied sequentially to the two nodes identified by the $\color{red}*$ symbols.}
\end{figure}

\section{Simplification of HaPPY code}
\label{ap:Happy}
\begin{figure}[H]
    \centering
    \includegraphics[width=5.2in]{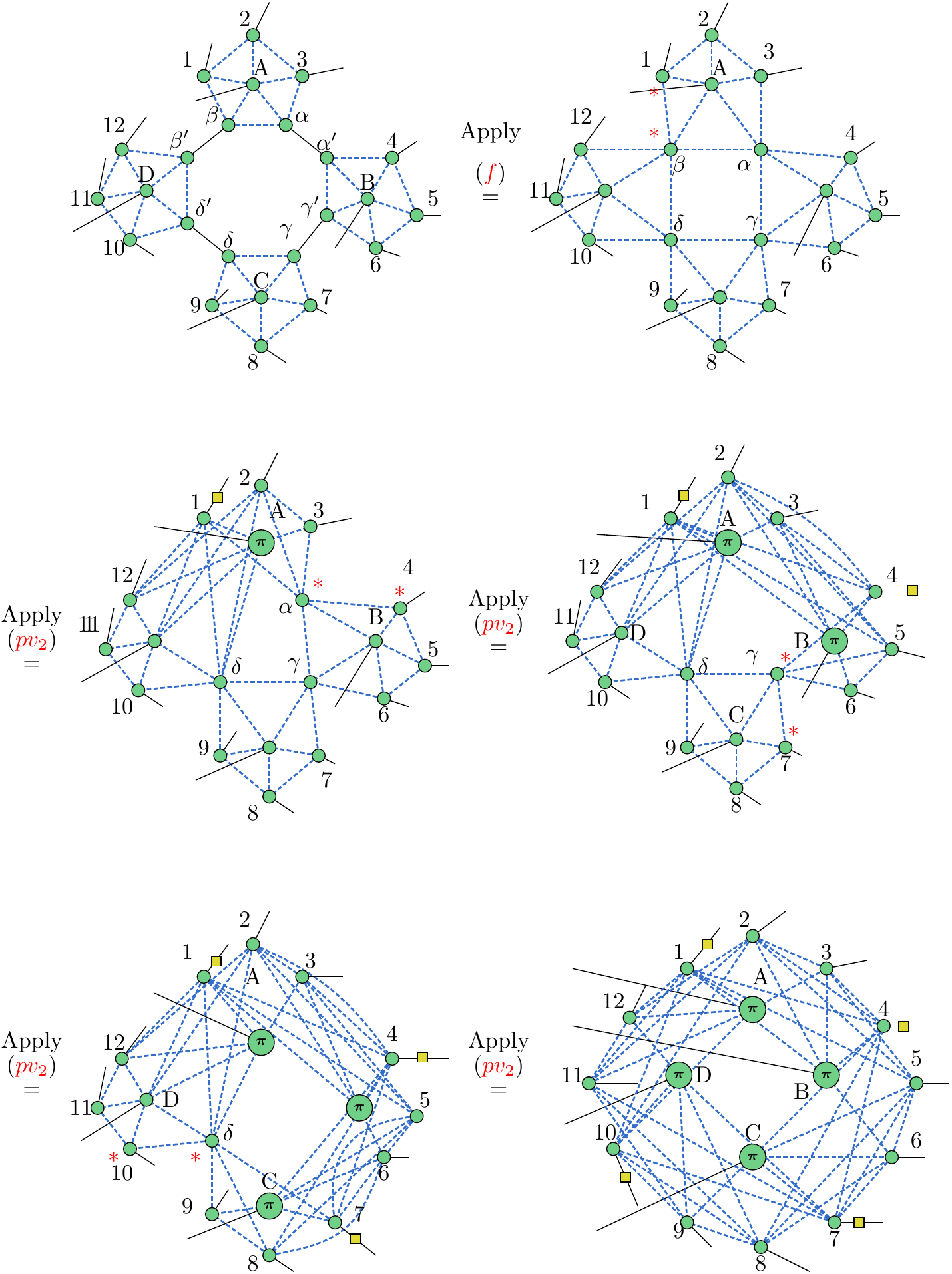}
    \caption{Construction of the HaPPY code from the ZX-calculus encoder diagram. The operation ($\color{red}pv_2$) is applied sequentially to the two nodes identified by the $\color{red}*$ symbols.}
    \label{fig:holo_step_2}
\end{figure}

\bibliography{CodeConcatenation}

\end{document}